\documentclass[draftcls, onecolumn,12pt]{IEEEtran}
\normalsize
\usepackage{color}
\usepackage{bbm}
\usepackage{gensymb}
\usepackage{extarrows}
\usepackage{url}
\usepackage{balance}
\usepackage{graphicx}
\usepackage[outdir=./]{epstopdf}
\usepackage{stfloats} 
\usepackage{amsfonts}
\usepackage{cite}
\usepackage{graphicx}
\usepackage{amsmath}
\usepackage{times}
\usepackage{latexsym}
\usepackage{subfigure}
\usepackage{bm}
\usepackage{amssymb}
\usepackage{cases}
\usepackage{array}
\usepackage{flushend}
\usepackage{fancyhdr}
\usepackage{stfloats}
\usepackage{setspace}
\usepackage{flushend}
\usepackage{mathrsfs}
\usepackage{amsthm} 
\usepackage{multirow}
\usepackage{lipsum}
\usepackage{extarrows}
\usepackage{epsfig}
\usepackage{multicol}
\usepackage{subfigure}
\usepackage[square, comma, sort&compress, numbers]{natbib}
\newtheorem{theorem}{Theorem}
\newtheorem{lemma}{Lemma}

\newtheorem{assumption}{Assumption}
\newtheorem{corollary}{Corollary}
\newtheorem{proposition}{Proposition}

\title{MmWave vehicle-to-infrastructure communication: Analysis of
	urban microcellular networks}

\author{\thanks{
	``Copyright (c) 2015 IEEE. Personal use of this material is permitted. However, permission to use this material for any other purposes must be obtained from the IEEE by sending a request to pubs-permissions@ieee.org. "}\IEEEauthorblockN{Yuyang Wang, ˜\IEEEmembership{Student Member, IEEE,}\thanks{Yuyang Wang, Kiran Venugopal, and Robert W. Heath Jr. are with the Department of Electrical and Computer Engineering, the University of Texas at Austin, Austin, TX, 78712 USA, email: $\{$yuywang, kiranv, rheath$\}$@utexas.edu.}\thanks{
			Andreas F. Molisch is with the Department of Electrical Engineering, University of Southern California, Los Angeles, CA, 90089-2565 USA, email: $\{$molisch$\}$@usc.edu.}
	}\thanks{Part of this work has been presented in IEEE VTC fall 2016 \cite{WanVenMol:Analysis-of-Urban-Millimeter:16}. This research was partially supported by the U.S. Department
		of Transportation through the Data-Supported Transportation
		Operations and Planning (D-STOP) Tier 1 University
		Transportation Center, and by the Texas Department
		of Transportation under Project 0-6877 entitled ``Communications
		and Radar-Supported Transportation Operations and
		Planning (CAR-STOP)", and a gift from Huawei. The work of A. F. Molisch was supported by
		the National Science Foundation and a gift from Samsung.
	}Kiran Venugopal, \IEEEmembership{Student Member, IEEE}, \\Andreas F. Molisch, \IEEEmembership{Fellow, IEEE}, and Robert W. Heath Jr., \IEEEmembership{Fellow, IEEE} }

\begin{document}

\maketitle
	\begin{abstract}
		Vehicle-to-infrastructure (V2I) communication may provide high data rates to vehicles via millimeter-wave (mmWave) microcellular networks. This paper uses stochastic geometry to analyze the coverage of urban mmWave microcellular networks. Prior work  used a  pathloss model with a line-of-sight probability function based on randomly oriented buildings,  to determine whether a link was line-of-sight or non-line-of-sight. In this paper, we use a pathloss model inspired by measurements, which uses a Manhattan distance pathloss model and accounts for differences in pathloss exponents and losses when turning corners. In  our model, streets are randomly located as a Manhattan Poisson line process (MPLP) and the base stations (BSs) are distributed according to a Poisson point process. Our model is well suited for urban microcellular networks where the BSs are deployed at street level. Based on this new approach, we derive the coverage probability under certain BS association rules to obtain closed-form solutions without much complexity. In addition, we draw two main conclusions from our work. First, non-line-of-sight BSs are not a major benefit for association or source of interference most of the time. Second, there is an ultra-dense regime where deploying active BSs does not enhance coverage.
	\end{abstract}
	\IEEEpeerreviewmaketitle
	
	
	
	\section{Introduction}\label{sec:intro}	
	Vehicle-to-infrastructure (V2I) communication offers the potential to enhance safety and efficiency in urban vehicular networks \cite{Mil:Vehicle-to-vehicle-to-infrastructure-V2V2I-intelligent:08,BelValPai:On-wireless-links-for-vehicle-to-infrastructure:10,GozSepBau:IEEE-802.11p-vehicle:12}. Combined with millimeter wave (mmWave) \cite{PiKha:An-introduction-to-millimeter-wave-mobile:11,RapGutBen:Broadband-millimeter-wave-propagation:13,RapHeaDan:Millimeter-wave-wireless:14}, V2I has the potential to offer high data rates and low latency \cite{ChoVaGon:Millimeter-Wave-Vehicular-Communication:16,CheShaZhu:Infotainment-and-road-safety:11,AleSanCui:Measurement-and-Analysis-of-Propagation:08}, enabling massive data sharing among a great number and diversity of mobile devices in vehicular networks \cite{ChoVaGon:Millimeter-Wave-Vehicular-Communication:16,VaShiBan:Millimeter-Wave-Vehicular:16}. MmWave communication not only has access to larger bandwidths, it can also allow large yet very compact antenna arrays deployed at both the transmitter and receiver to provide high directional beamforming gains and low interference.  {Compared to channels at microwave frequencies ($<$6 GHz), however, mmWave channels are more sensitive to blockage losses, especially in urban streets where signals are blocked by high buildings, vehicles or pedestrians \cite{AleSanCui:Measurement-and-Analysis-of-Propagation:08}, \cite{NiuLiJin:A-survey-of-millimeter-wave:15}, and sharp transitions from line-of-sight (LOS) to non-line-of-sight (NLOS) links are more common.}
	This motivates the study of mmWave microcellular network performance in the context of vehicular urban areas. 
	
	\subsection{Related Work}
	
	\textbf{Urban street model:} {Stochastic geometry has been used extensively to analyze performance in mmWave cellular  networks \cite{AndBacGan:A-tractable-approach-to-coverage:11,BaiHea:Coverage-and-rate-analysis:15,BacZha:A-correlated-shadowing-model:15, AndBaiKul:Modeling-and-Analyzing-Millimeter:17,ElsKulBoc:Downlink-and-uplink-cell:16,Di-:Stochastic-geometry-modeling:15}. BS and cellular user locations are modeled as Poisson point processes on a two-dimensional plane, based on which the coverage probablity of a \emph{typical} cellular user is derived. Also, building blockages are considered as the main source differentiating LOS and NLOS links, with a few papers analyzing different building blockage models.}
	Unfortunately, prior work analyzing mmWave cellular networks in \cite{AndBacGan:A-tractable-approach-to-coverage:11,BaiHea:Coverage-and-rate-analysis:15,BacZha:A-correlated-shadowing-model:15, AndBaiKul:Modeling-and-Analyzing-Millimeter:17,ElsKulBoc:Downlink-and-uplink-cell:16,Di-:Stochastic-geometry-modeling:15} employed a pathloss model with a LOS probability function based on Euclidean distance \cite{BaiVazHea:Analysis-of-blockage-effects:14}, to determine whether a link was LOS or NLOS. This works well for randomly oriented buildings \cite{BaiHea:Coverage-and-rate-analysis:15}, but does not properly model V2I networks where strong LOS interference may result from infrastructure co-located on the same street. 
	
	Recent work has considered alternative topologies that may better model urban areas. 
	{In \cite{KulSinAnd:Coverage-and-rate-trends:14}, an approach to determine LOS and NLOS BSs by approximating a \emph{LOS ball} was proposed. The model was shown to be able to better approximate the LOS area than \cite{BaiVazHea:Analysis-of-blockage-effects:14}. 
		In \cite{LeeZhaBac:A-3-D-Spatial-Model-for-In-Building:16}, three-dimensional Poisson buildings were modeled using
		Poisson processes to characterize the correlated shadowing
		effects in urban buildings. The idea was to add one more
		dimension to the Manhattan Poisson line processes (MPLP), by modeling the floor locations as Poisson
		process. This allowed an exact characterization of coverage of
		indoor urban cellular networks.}  In \cite{BacZha:A-correlated-shadowing-model:15}, a stochastic geometry model in a Manhattan type network was analyzed, since it is a tractable yet realistic model for Manhattan type urban streets. The urban streets were modeled as one-dimensional MPLP and the coverage probability was derived considering the penetration of signal through buildings. {Unfortunately, the results in \cite{BacZha:A-correlated-shadowing-model:15} used a pathloss model mainly considering the penetration effects of signals through urban buildings, with a fixed loss for each penetration. This is not applicable for mmWave systems where penetration loss is high. }In this paper, we also use the MPLP for modeling the urban street distribution, but combined with a mmWave-specific channel model. 
	
	\textbf{Urban mmWave channel model:} There is a vast body of literature concerning mmWave channel modeling in urban areas, see, e.g., \cite{MolKarWan:Millimeter-wave-channels-in-urban:16} and references therein. One of the key characteristics of urban environment is the high density of streets and high-rise buildings. {Since mmWave signals are  sensitive to blockage, which induces significant signal attenuation}, LOS and NLOS links can have sharply different pathloss exponents, as was also shown in numerous measurements \cite{RapSunMay:Millimeter-wave-mobile:13}, \cite{MacZhaNie:Path-loss-models:13}, \cite{RapGutBen:Broadband-millimeter-wave-propagation:13}, and is reflected in the standardized channel models \cite{NurKarRoi:METIS-channel-models:15}.  Investigations in a variety of environments showed that, in general, penetration loss increases with carrier frequency.  For modern buildings with steel concrete and energy saving windows, in particular, penetration through just one wall can incur losses in the order of 30 dB; therefore, propagation {\em through} buildings is not a relevant effect in mmWave Manhattan type urban environments \cite{HanZhaTan:5G-3GPP-like-channel-models:16}.
	
	In \cite{KarMolHur:Spatially-Consistent-Street-by-Street:17},  a spatially consistent pathloss model was proposed for urban mmWave channels in microcells. Based on ray tracing, it was shown that the pathloss exponents differ from street to street and should be modeled as a function of both the street orientation and the absolute location of the BS and user equipment (UE)\footnote{Henceforth we assume a downlink so that receiver and UE can be used exchangeably.}. Hence, the signal is seen as propagating along different streets, with diffraction effects happening at the corner, instead of penetrating through the urban buildings. The pathloss is summed up by the individual pathloss on different segments of the propagation paths,  incorporating an additional loss at each corner. This shows that the  Euclidean distance might not be a good measure to characterize the pathloss effects in urban microcell networks at mmWave.  In this paper, we adopt a modified pathloss model  similar to \cite{KarMolHur:Spatially-Consistent-Street-by-Street:17} based on the Manhattan distance, which enables tractable analysis while still retaining the key features of the mmWave microcellular channel.

	\subsection{Contributions}
	
	In this paper, we develop a tractable framework to characterize the downlink coverage performance of urban mmWave vehicular networks. Specifically, we consider snapshots of the urban microcellular network, without modeling vehicle mobility. This reduces the network to an urban mmWave microcellular networks. We model the location of urban streets by a MPLP. The width of the street is neglected, and herein the blockage effects of vehicles are not considered in the analysis. We extend our previous paper \cite{WanVenMol:Analysis-of-Urban-Millimeter:16} to account for large antenna arrays and directional beamforming at mmWave. We use a modification of the sectorized antenna model for tractable analysis \cite{BaiHea:Coverage-and-rate-analysis:15}\cite{VenValHea:Device-to-Device-Millimeter-Wave:16} and apply the new pathloss model from \cite{KarMolHur:Spatially-Consistent-Street-by-Street:17}. The pathloss model is characterized by the Manhattan distance of the propagation link, which, with MPLP street modeling, yields tractable results for coverage analysis. 
	
	Based on our model, we analyze coverage of randomly located UEs on the roads formed by the lines, which is different from the conventional approach where coverage is analyzed conditioned on the links being outdoors \cite{BaiHea:Coverage-and-rate-analysis:15}. We adopt a new procedure in the calculation of coverage probability, compared to the previous work \cite{BaiHea:Coverage-and-rate-analysis:15}.  We analyze the coverage probability by first computing the cumulative distribution function (CDF) of associated BS link gain and then the coverage probability conditioned on the associated link gain. By averaging over the conditioned received signal power, we obtain simple but accurate expression of coverage probability. 
	
{Compared to \cite{WanVenMol:Analysis-of-Urban-Millimeter:16}, this paper also includes the following contributions. Based on the coverage probability, we obtain insights concerning the scaling laws of coverage probability with street and BS intensities, the sensitivity of coverage to the channel conditions and the effects of LOS/NLOS interference. Also, we derive closed-form expressions of the LOS BS association probability, under different channel conditions. We then use the map data of the streets in Chicago from \emph{OpenStreetMap} \cite{url_openstreetmap,HakWeb:Openstreetmap:-User-generated-street:08,RamTopChi:OpenStreetMap:-using-and-enhancing:11} and extract it using the Geographical Information System (GIS) application \emph{QGIS} \cite{QGI:Quantum-GIS-geographic-information:11}. This is used to compare the ergodic achievable rate of realistic streets, MPLP street model and fixed grid models. The comparison shows that the MPLP based analysis is valid for outdoor microcell urban networks at mmWave.  
}
	
	\section{System Model}\label{sec:system_model}
	In this section, we explain the key assumptions and models adopted in this paper. First, we explain the street model in urban vehicular networks. Then, we present a tractable form of the pathloss model of mmWave microcells based on Manhattan distance from \cite{KarMolHur:Spatially-Consistent-Street-by-Street:17}. We introduce a modified mmWave sectorized  antenna pattern that is used for our analysis. Lastly, we formulate the signal-to-interference-plus-noise ratio (SINR) of the receiver and demonstrate the rule of the strongest propagation path.
	\subsection{Network model-MPLP}\label{sec:mplp}

	\begin{figure}[h!]
		\centering
		\setlength{\abovecaptionskip}{0pt}
		\setlength{\belowcaptionskip}{0pt}
		\includegraphics[width = 5in]{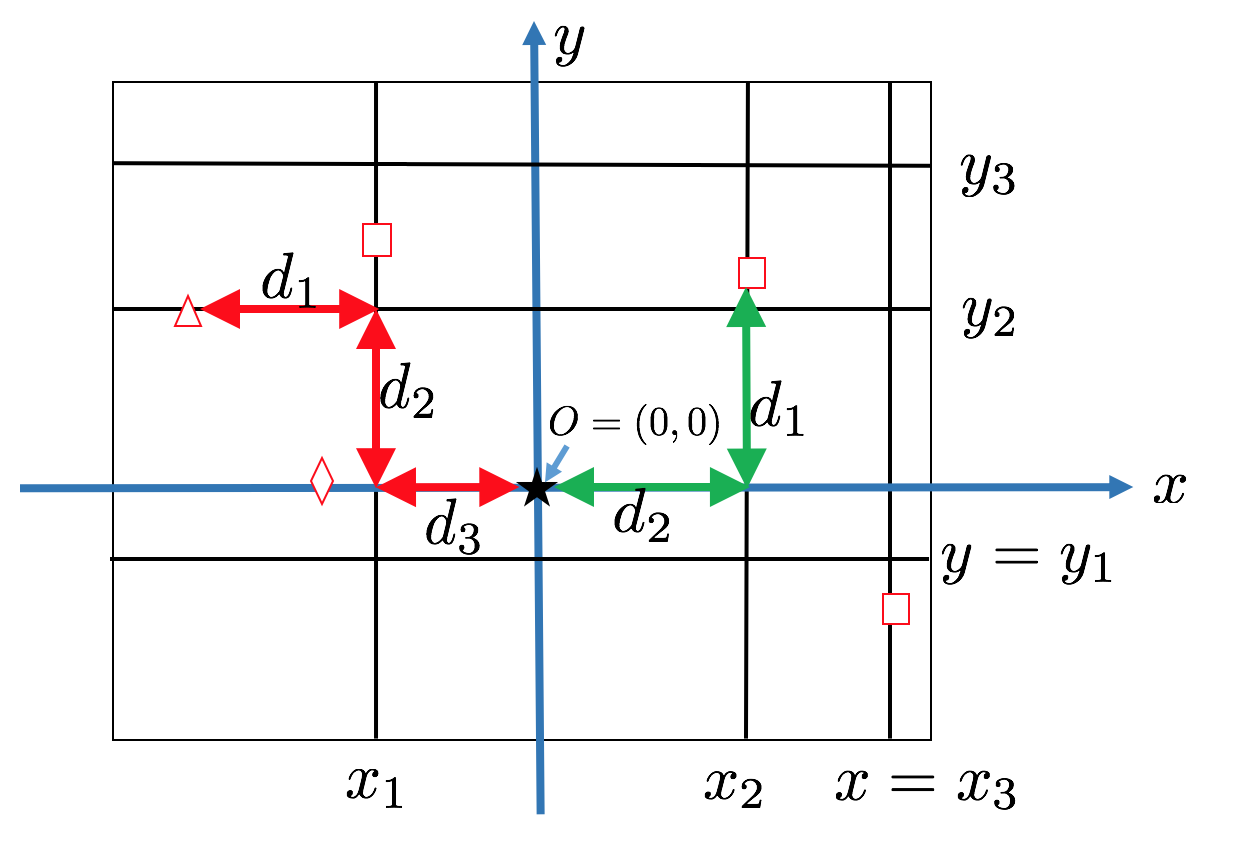}
		\caption{An illustration of our proposed pathloss model under the Cartesian coordinate system. In the Cartesian coordinate, the \emph{cross} streets (parallel to the $y$-axis) are represented by $x= x_i$, where the intercept $x_i$ denotes the location of the street. Similarly, the \emph{parallel} streets (parallel to the $x$-axis) are denoted by  $y= y_i$. The $\bigstar$ is the typical receiver (also the origin of the coordinate), the diamond $\diamond$ represents one typical BS, $\square$ is a cross BS and $\triangle$ is a parallel BS. The red line denotes one propagation link of a parallel BS and the green line is one propagation path of a cross BS. The pathloss in decibel scale is added up by the pathloss on each individual segments of the propagation path.}\label{mplp_fig}
	\end{figure}  
	
	We show in Fig.~\ref{mplp_fig} an illustrative snapshot of the Manhattan network in a Cartesian coordinate system.  Without loss of generality, a \emph{typical} receiver is placed at the origin $O$, and the streets are assumed to be either perfectly horizontal or vertical in the coordinate system. We call the street where the receiver is located at as the \emph{typical} street, i.e., the $x$-axis. We refer to other  horizontal and vertical streets respectively as parallel streets and cross streets. These streets are generated from two independent one-dimensional homogeneous Poisson point processes (PPP) ${\Psi}_x$ and $\Psi_y$, with identical street intensity $\lambda_\mathrm {S}$. Under the current coordinate system, we define the set of \emph{parallel} streets as $\bigcup_{y_i\in\Psi_y} L_y(y_i) $, where $\bigcup$ denotes union of sets, and $L_y(y_i)$ denotes the parallel street with intercept (location) at $y_i$. The set of the \emph{cross} streets is defined as $\bigcup_{x_j\in\Psi_x} L_x(x_j)$, with $L_x(x_j)$ similarly defined as the cross street having intercept $x_j$ on the $x$-axis. By Slivnyak's theorem [15], [33], the \emph{typical} street $y =0$ is added to the process. BSs are deployed at the street level, and are distributed on each cross, parallel, and the typical street as independent one-dimensional homogeneous PPPs. Similar to the naming convention for the streets, we name the BSs on the typical streets as \emph{typical} BSs, and cross BSs and parallel BSs on the cross and parallel streets respectively. 	
	\subsection{Pathloss model}\label{ssec:sys_model:pathloss}	
	
	We adopt a pathloss model that is based on the Manhattan distance instead of Euclidean distance. The model is similar to  \cite{KarMolHur:Spatially-Consistent-Street-by-Street:17}, but uses several modifications to provide tractability. Ray tracing shows that in an urban mmWave microcell,  Euclidean distance might not be a dominant parameter in pathloss modeling. Since the penetration through urban building walls is negligible at mmWave, the signal detours its way along the streets in urban canyons and changes its directions by diffractions on the buildings at intersections. Therefore, instead of the direct Euclidean distance between the BS and the receiver, the street orientation relative to the BS location, and the absolute positions of the BS and receiver are the key parameters to determine the pathloss. 
	
	It is shown by the ray tracing results that a way to model the net pathloss of a propagation link in urban mmWave microcells is to add up the pathloss on different segments of the propagation paths, with an additional loss when the waves couple into a new street canyon. The propagation path may be thought of as \emph{segmented}, with the signal changing directions to find LOS paths, circumventing building blockage. We assume that there are in total of $M$ segments along the propagation paths, i.e., $M-1$ corners where signal changes directions. Note that the value of $M$ depends on the actual position of the BS and the receiver.  The individual length of the $i$-th segment is denoted as $d_i$, the pathloss exponent on the $i$-th segment is $\alpha_i$, and the corner loss at the corner of the $i$th street segment and $i+1$-th segment is $\Delta$ (in decibel scale), where we assume corner losses at different corners are identical. 
	
	We define the \emph{LOS segment} as the first segment of the propagation path from the BS and \emph{NLOS segment} as the remaining segments on the propagation path. It should be noted that the LOS and NLOS defined for the segments are only indicating the \emph{order} of different segments of the propagation paths, which is different from the definition of LOS/NLOS \emph{paths} in traditional representations. 
	We assume that LOS segments on different streets share the same pathloss exponent $\alpha_\mathrm L$, while the pathloss exponent for NLOS segments is $\alpha_\mathrm N$. Notice that the equation is \emph{not ``symmetric"}, i.e., the street segment that has LOS to the BS has a pathloss coefficient that is different from the one that has LOS to the UE; such a situation might occur due to the different heights of UE and BS. To clarify, the pathloss \emph{does not} hold for vehicle-to-vehicle channel modeling, since the model is asymmetric. To conclude, the pathloss in the decibel scale is defined as follows 
	\begin{align}\label{equ:pathloss}                  
		&\mathrm{PL}_\text{dB} =
		10\left( \alpha_\mathrm L \log_{10}d_1 + \alpha_\mathrm N\sum_{i=2}^{M}\log_{10}d_i\right) + (M-1)\Delta. 
	\end{align}
	With this Manhattan distance based pathloss model,  we can classify the BSs into three categories, as illustrated in Fig.~\ref{mplp_fig}:  i) BSs on the typical street (\emph{typical BSs}) that have one direct propagation path to the typical receiver; ii) NLOS BSs on the cross streets (\emph{cross BSs}) that have a propagation path consisting of a LOS segment (green path $d_1$) and NLOS segment (green path $d_2$) to the typical receiver, and iii) NLOS BSs on the parallel streets (\emph{parallel BSs}) that have a propagation path consisting of a LOS segment (red path $d_1$) and two NLOS segments (red path $d_2$, $d_3$). The analysis of the strongest path of different BSs will be provided in Section \ref{sec:strongest}. This pathloss model also bears a strong relationship to \cite{Ber:A-recursive-method-for-street:95}, which considered the pathloss model in urban microcells where waves are coupled at the street corners with different angles. 
	\subsection{Sectorized antenna model}\label{sec:antenna}
	To leverage array gain, directional beamforming by multiple antennas is performed at mmWave BSs. For simplicity, we assume the receiver has an omni-directional antenna, and the BSs are equipped with $N_\mathrm t$ transmit antennas. We adopt a sectorized antenna model for the BS \cite{BaiHea:Coverage-and-rate-analysis:15}, \cite{VenValHea:Device-to-Device-Millimeter-Wave:16}, with the main lobe gain denoted as $G$ and the side-lobe gain as $g$.  The beamwidth of the main lobe is $\theta$, as shown as the \emph{red fan} in Fig. \ref{fig:beam_pattern} and all the other directions outside the main lobe are assumed to be in the side lobe (shown in the blue circle). 
	For a uniform planar antenna array, the main lobe gain can be approximated by $G = N_\mathrm t$, which is  the maximum power gain that can be supported with $N_\mathrm t$-element antenna array. The side-lobe gain is evaluated by $g = \frac{\sqrt{N_\mathrm t}-\frac{\sqrt{3}}{2\pi}N_\mathrm t\sin\left(\frac{\sqrt{3}}{2\sqrt{N_\mathrm t}}\right)}{\sqrt{N_\mathrm t} - \frac{\sqrt{3}}{2\pi}\sin\left(\frac{\sqrt{3}}{2\sqrt{N_\mathrm t}}\right)}$, which is calculated to satisfy the following antenna equation for constant total radiated power \cite{VenValHea:Device-to-Device-Millimeter-Wave:16}, \cite{Bal:Antenna-Theory:-Analysis:16},
	\begin{align}
	\int_{-\pi}^\pi \int_{-\frac{\pi}{2}}^\frac{\pi}{2} G(\phi, \psi)\cos(\psi)d\psi d\phi = 4\pi,
	\end{align} 
	and the beamwidth is $\theta =   \frac{\sqrt{3}}{\sqrt{N_\mathrm t}}$.
	\begin{figure}[h!]
		\centering
		\setlength{\abovecaptionskip}{0pt}
		\setlength{\belowcaptionskip}{0pt}
		\includegraphics[width = 5in]{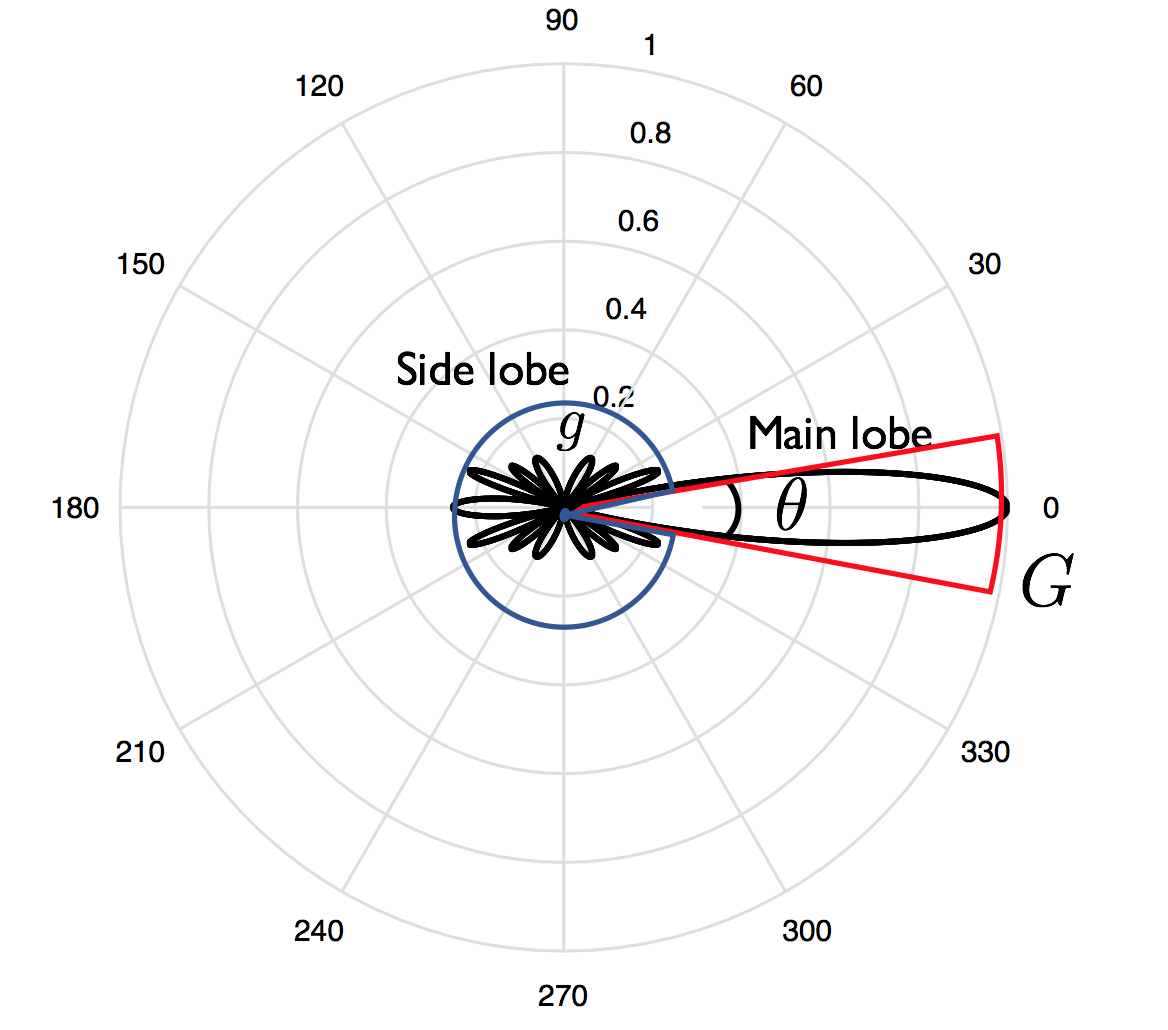}
		\caption{An example illustration of the simplified sectorized antenna pattern. We only consider the main lobes and the side lobes. Main lobes and side lobes are assumed to have identical gain on different directions, respectively denoted by $G$ and $g$.  }\label{fig:beam_pattern}
	\end{figure}

	\subsection{Signal-to-interference-plus-noise ratio (SINR) }\label{sec:cov}
	SINR coverage analysis is important to determine outage holes and ergodic throughput of wireless networks. While these metrics in the context of mmWave-based vehicular networks depend on both mobility and the blockage effects due to the vehicles, in this paper we simply consider snapshots of the urban microcelluar network and look at the distribution of the instantaneous SINR. This approach is taken to confirm the analytic tractability of the pathloss model described in Section \ref{ssec:sys_model:pathloss}, which captures the blockage and shadowing effects due to buildings and accounts for the geometry of streets in an urban environment. In this section, we will explain the key assumptions of BS association and the definition of the interference. 
	\subsubsection{BS association}\label{subsec:asso}
			In our model, as mentioned in Section \ref{sec:antenna}, we assume the BSs deploy directional beamforming to exploit antenna gain, while at the receiver side, the antenna is omni-directional. During the cell discovery and BS association process, we assume all BSs do exhaustive beam search over the entire beam space by beam sweeping, each at an individual time slot. Based on the reference signal received power (RSRP) of each beam, the receiver can determine the serving BS and the associated beam by selecting the strongest RSRP. After exhaustive beam sweeping, \emph{the receiver is always aligned with the main lobe}, therefore the antenna gain of the associated BS is $G$ all the time. Therefore, the receiver is simply associated to the BS with the smallest pathloss defined in (\ref{equ:pathloss}), without including extra antenna gain.
\subsubsection{Interference}
From the BS association rule, the receiver is associated to the BS with the smallest pathloss, i.e., the largest path gain, which we denote as $u$.  Therefore, interference arises from other BSs whose path gains are smaller than $u$, with an extra beamforming gain $\mathcal{G}$ added on. Given the orientation of the main beam (towards the desired user), other BSs could either point the main lobe or the side lobe towards the referenced (typical) receiver, based on the sectorized antenna model. Therefore, the beamforming gain of the interference $\mathcal{G}$ is random, and is represented as 
	\begin{align}\label{equ:rbf}
	\mathcal{G} = A G + (1-A) g, 
	\end{align}
	where $A = \mathbb{I}(p)$ and $p$ is the probability that the interference link from the BS has beamforming gain of $G$, and $\mathbb{I}(\cdot)$ is the Bernoulli function. 

\subsubsection{Formulation of SINR}

	We use $\Phi_\mathrm T$ to denote the set of LOS link distances $x_\mathrm T$ from the  typical BSs to the receiver. The set of lengths of the horizontal and vertical links, $x_\mathrm C$ ($d_1$ in green, Fig. 1) and $y_\mathrm C$ ($d_2$ in green, Fig. 1), constituting the propagation path from the cross BSs is denoted as ${\Phi_\mathrm C}$. Similarly, $\Phi_\mathrm P$ is used to denote the set of distances $(x_\mathrm P, y_\mathrm P,z_\mathrm P)$ ($d_3$, $d_2$, $d_1$ in red) corresponding to the propagation path from parallel BSs (see Fig.~\ref{mplp_fig}). To simplify the demonstration, we define the path gain of the LOS and NLOS segment respectively as 
	\begin{align}
		\ell_\mathrm L (x)  =  x^{-\alpha_\mathrm L},
	\end{align}
	and 
	\begin{align}\label{ellnlos}
		\ell_\mathrm N (x)  = cx^{-\alpha_\mathrm N},
	\end{align}
	where $x$ is the length of the propagation segment. The corner loss term $c = 10^{-\Delta/10}$ in the total path gain expression is also captured along with the propagation loss associated with each NLOS segment in  (\ref{ellnlos}), with $\alpha_\mathrm N$ denoting the NLOS pathloss exponent.
	  
 We denote $h_o$ as the small scale fading of the typical receiver $o$ from the associated BS and $h_i$ as the small scale fading of the $i$th BS in the Poisson point processes. $\mathcal{G}_i$ is a beamforming gain associated with each interfering BS, defined in (\ref{equ:rbf}), and $N_0$ is defined as the noise variance. $\Phi_\mathrm T'$, $\Phi_\mathrm C'$ and $\Phi_\mathrm P'$ are the set of segment lengths of the interfering BSs. 	Conditioning on the associated BS link gain $u$ (which includes both the path gain $\mathrm{PL}$ in (\ref{equ:pathloss}) and the antenna beamforming gain, which is always $G$), the SINR can be formulated as follows, in terms of interference components, respectively from the typical BSs $I_{\phi_\mathrm T}$, cross BSs $I_{\phi_\mathrm C}$ and parallel BSs $I_{\phi_\mathrm P}$,
	\allowdisplaybreaks
	\begin{align}
		\mathrm{SINR} &= \frac{h_o u}{N_0 + I_{\phi_{\mathrm T}}(o) + I_{\phi_{\mathrm C}}(o) + I_{\phi_{\mathrm P}}(o)}, \label{(sinr_los)} \\
		\text{with}~
		I_{\phi_{\mathrm T}}(o) &= \sum_{x^i_\mathrm T\in\Phi_\mathrm T'}\mathcal{G}_ih_i \ell_\mathrm L (x^i_\mathrm T),\label{(i_phi1)} \\
		I_{\phi_{\mathrm C}}(o) &= \sum_{(x^i_\mathrm C,y^i_\mathrm C)\in\Phi_\mathrm C'}{\mathcal{G}_ih_i \ell_\mathrm N(x^i_\mathrm C)} \ell_\mathrm L(y^i_\mathrm C), \label{i2} \\
		\text{and}~
		I_{\phi_{\mathrm P}}(o) &= \sum_{(x^i_\mathrm P, y^i_\mathrm P, z^i_\mathrm P)\in\Phi_\mathrm P'}\mathcal{G}_i h_i\ell_\mathrm N(x^i_\mathrm P)
		\ell_\mathrm N(y^i_\mathrm P)\ell_\mathrm L(z^i_\mathrm P). \label{i3}
	\end{align}	
	Based on the assumption in Section \ref{subsec:asso}, and conditioning on the associated BS path gain as $u$, we have the following constraints for the sets of interfering BSs' segment lengths $\Phi_\mathrm T', \Phi_\mathrm C'$ and $\Phi_\mathrm P'$ in (\ref{(i_phi1)}) -- (\ref{i3}) as
	\begin{align}
		\Phi_\mathrm T' &= \{x_\mathrm T\in\Phi_\mathrm T\Big| \ell_\mathrm L(x_\mathrm T)  <  u\}, \label{(los_constraint_1)} \\
		\Phi_\mathrm C' &= \{\left(x_\mathrm C,y_\mathrm C\right)\in\Phi_\mathrm C \Big| \ell_\mathrm N(x_\mathrm C)\ell_\mathrm L(y_\mathrm C) <  u\}, \label{(los_constraint_2)} \\
		\text{and}~ 
		\Phi_\mathrm P' &= \{\left(x_\mathrm P,y_\mathrm P, z_\mathrm P\right)\in\Phi_\mathrm P\Big| \ell_\mathrm N(x_\mathrm P)
		\ell_\mathrm N(y_\mathrm P)\ell_\mathrm L(z_\mathrm P) <  u\}. \label{(los_constraint_3)}
	\end{align}
	The above constraints are based on the assumption that perfect beam sweeping is done for each surrounding BS in the initial access, which leads to (\ref{(i_phi1)}) -- (\ref{i3}). 
	\subsection{Analysis of strongest path}\label{sec:strongest}
Given a BS at fixed location (either a typical, cross or parallel BS), the received power from the BS is still not clear, even though we have already defined our pathloss model in (\ref{equ:pathloss}). This is because, first, the Manhattan pathloss model bears huge differences with the Euclidean distance based pathloss model. Secondly, given a BS location, there could be multiple \emph{paths} for the signal to reach the receiver within the grid-type Manhattan city. Since we assume the antenna pattern at the BS is sectorized, there exists radiated power to all directions, with different antenna gains With different paths routed for a signal radiated from all directions, the received power comes from different paths. To make it tractable, we make the following assumption for the analysis herein.  

	\begin{assumption}
		For analysis, we only consider the path from BS to the typical receiver with the largest received power. 
	\end{assumption} 
	To be the strongest path, the path should have i) shorter individual path segment lengths, ii) fewer individual segments, hence fewer corners and smaller corner loss, since pathloss is calculated by multiplying individual segment pathloss and one extra multiplication might reduce the path gain by orders of magnitude, iii) larger beamforming gain.
	\begin{figure}
		\centering	
		\includegraphics[width =  4.5in]{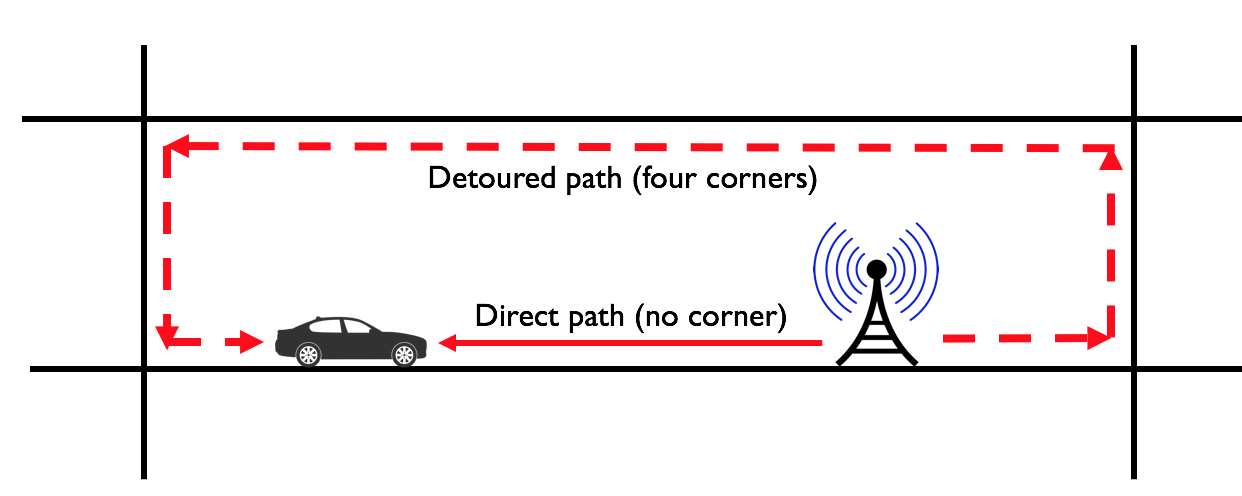}
		\caption{An illustration of the strongest path of a typical BS. There are many paths which the signal could follow from the BS to the typical receiver. The strongest two are respectively the direct path given in the solid red line, and the detoured path in red dashed line. There is a big difference between the direct and detoured path gain, due to the existence of the extra corner losses (four more corners) of the detoured path.}	\label{fig:typical}
	\end{figure}
	\begin{figure}
		\centering	
		\includegraphics[width=5in]{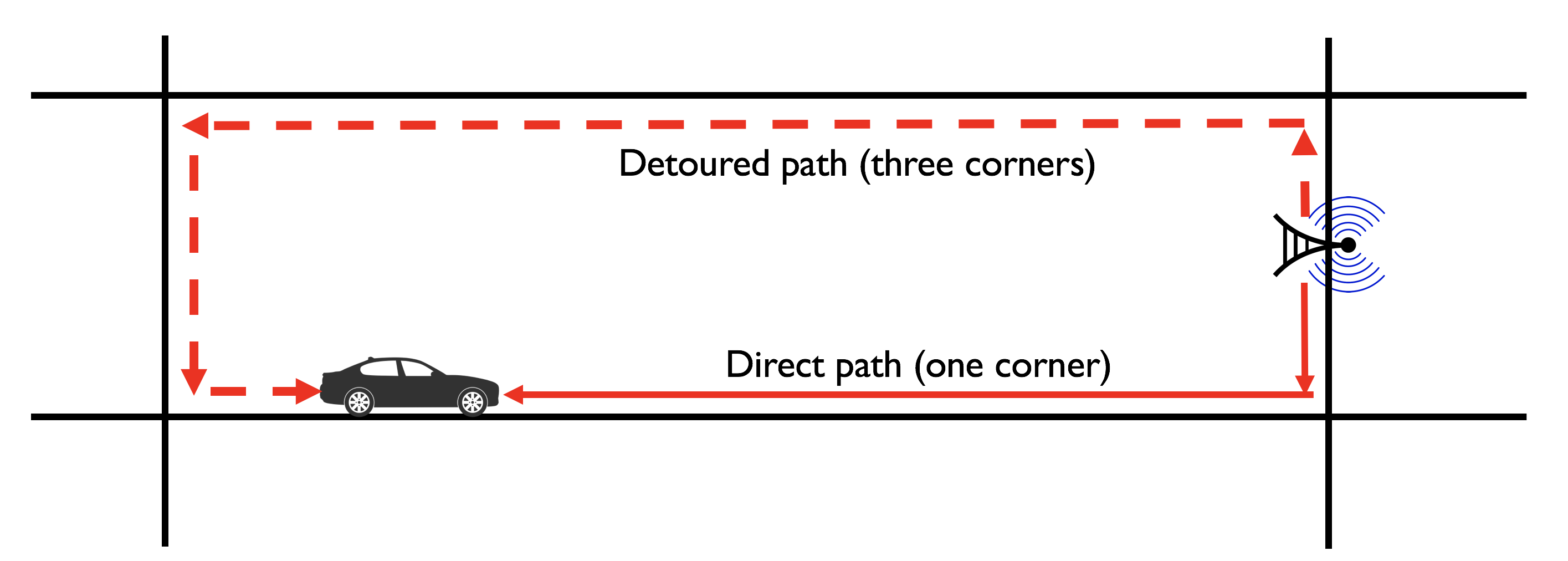}
		\caption{An illustration of the strongest path of a cross BS. Similarly as before, there exist two potential strongest path, denoted as the direct and detoured paths. The direct path has one corner loss, and the detoured path has three corner losses.}	\label{fig:cross}
	\end{figure}
	\begin{figure}
		\centering	
		\includegraphics[width=4in]{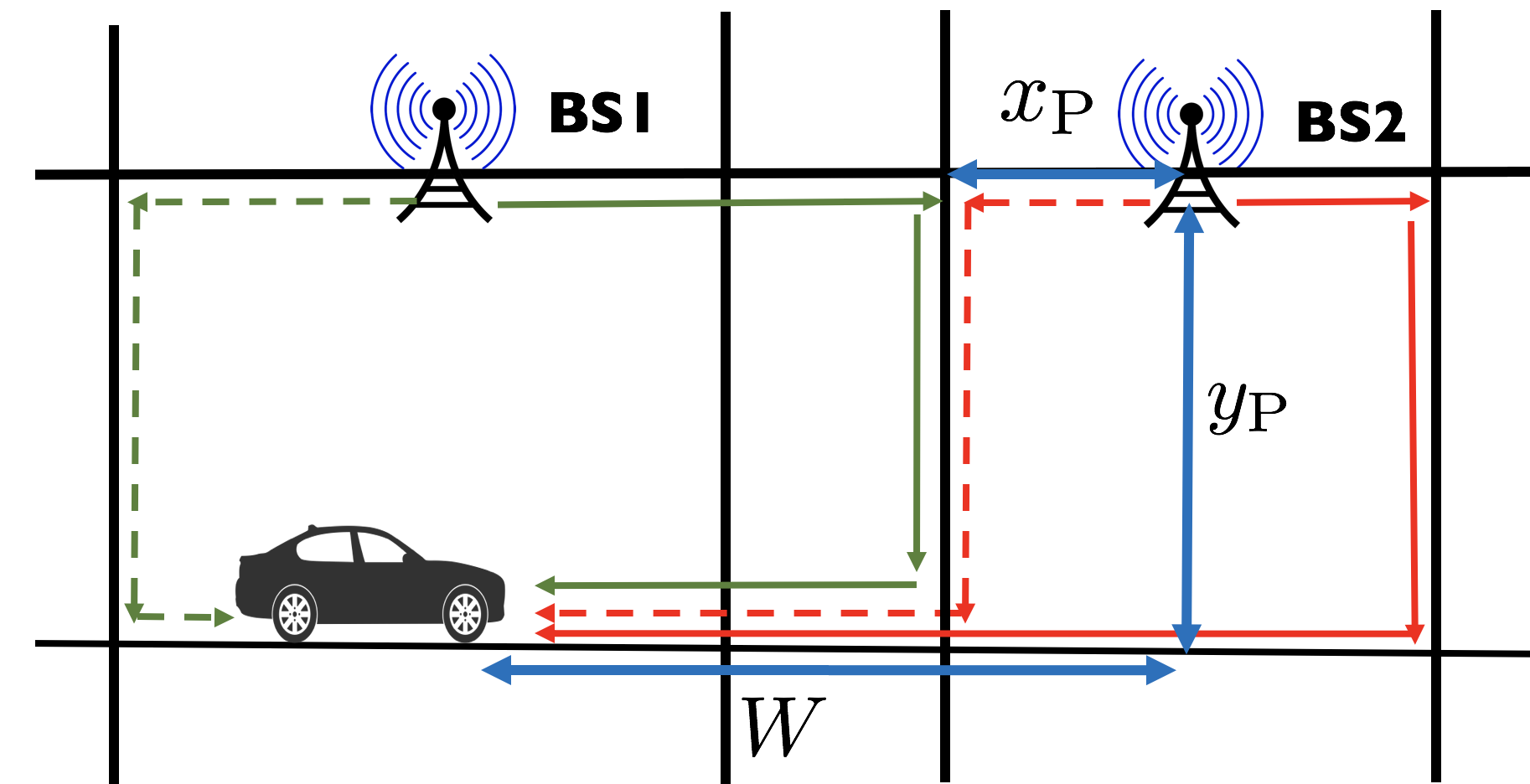}
		\caption{An illustration of the strongest path of a parallel BS. The \emph{candidate} strongest path is difficult to identify for the parallel BS simply based on the number of corner losses, since multiple paths have the same number of corner losses (two). There are also two types of parallel BSs, respectively the BS in the same block as the receiver, such as BS1; and those that do not lie in the same block as the receiver, e.g., BS2. The strongest paths for these two types of parallel BSs differ. }	\label{fig:parallel}
	\end{figure}

The strongest path for the BS association is simply the path with the smallest pathloss, since the receiver is always associated with the main lobe of the beam, with an identical beamforming gain as $G$. The strongest path for the interfering link analysis, however, is not necessarily the path with the smallest pathloss, due to the existence of  beamforming gain. Next, we demonstrate some different cases of the relative location of the receiver to the BS, in terms of the strongest receiver path. Fig.~\ref{fig:typical} and Fig.~\ref{fig:cross} illustrate the potential strongest paths of a typical and a cross BS. In each of the cases, there is one \emph{direct} path which has fewer corners and one \emph{detoured} path which detours its way before reaching the receiver. For typical BSs, the detoured path has \emph{four} more corners than the direct path; while for the cross BSs, there are \emph{two} more corners. Each corner introduces an approximately extra $20$dB loss, which is much more significant than the effects compensated by the beamforming gain difference. Therefore, even if the departure direction of the direct path lies inside the side lobe, the strongest path should still be the direct path.
	
	For the parallel BSs, both the detoured and direct paths have two corners, which makes it hard to identify the strongest path (see Fig.~\ref{fig:parallel}). In addition, the BSs are categorized to two types. It could either be in the \emph{same block} as the receiver (e.g., BS 1) or \emph{different block} as the receiver (e.g., BS 2), as shown in Fig.~\ref{fig:parallel}. For the same block BS,  the strongest path could either be the green dashed line or the green solid line. For the different block BS, however, the strongest path could traverse any of the cross streets and could point either left for right. 
	To make the analysis tractable, we make the following assumption. 
	
	\begin{assumption}
		For the strongest path of the parallel BSs, the signal travels along the LOS segment (first segment of the path) in the direction towards the receiver, rather than away from it. 
	\end{assumption}
	With this assumption, to find the strongest path for the parallel BS at different blocks as the receiver, we provide the following proposition. 
	
	\begin{proposition}\label{prop:strongestpath}
		The strongest propagation path from a parallel BS is via either the cross street $\Theta_{ \mathrm R}$ closest to the receiver or $\Theta_\mathrm B$ closest to the BS. 
		\begin{proof}
			Conditioning on the location of the parallel BS, the segment $y_\mathrm P$ and the corner loss $2\Delta$ of all propagation paths are the same, hence, the pathloss on the vertical link and the two corner losses can be taken out while formulating the following optimization problem. 
			
			For the interfering BS, since $\mathcal{G}$ is a random variable taking values of $G$ or $g$, we have $	\mathcal{G}\leq G$. Hence, the  maximum path gain of the parallel BS $G_\mathrm p$ can be upper bounded by 
			\begin{align}
				G_\mathrm P& \leq G - 2\Delta - 10 \alpha_\mathrm N \log_{10}y_\mathrm P + 10G_\mathrm M\nonumber\\
				&\leq G - 2\Delta - 10 \alpha_\mathrm N \log_{10}y_\mathrm P+ 10\max\left\{G_\mathrm M\right\}. 
			\end{align}
			where 
			\begin{align}
				G_\mathrm M =  -\alpha_\mathrm N \log_{10}x_\mathrm P - \alpha_\mathrm L \log_{10}z_\mathrm P,   
			\end{align}
			We then formulate the optimization problem of $G_\mathrm M$ as 
			\begin{equation}
				\begin{aligned}
					& \underset{x_\mathrm P, z_\mathrm P \in(0,W)}{\text{maximize}}
					& &-\alpha_\mathrm N \log_{10} x_\mathrm P - \alpha_\mathrm L \log_{10} z_\mathrm P \\
					& \text{subject to}
					& & x_\mathrm P + z_\mathrm P = W.
				\end{aligned}
			\end{equation}
			The objective function can be expressed as $P(x) = -\alpha_\mathrm{N} \log x -\alpha_\mathrm L \log (W-x)~ , x\in(0,W)$,  whose second order derivative is 
			\begin{align}
				P''(x) = \frac{\alpha_\mathrm N}{x^2} + \frac{\alpha_\mathrm L}{(W-x)^2}.
			\end{align} 
			The second order derivative of $P(x)$ is positive for all $\alpha_\mathrm{L}$, $\alpha_\mathrm{N}$, and $W$, which means $P(x)$ is convex. Denoting the distance from $\Theta_\mathrm R$ to the receiver as $x_1$ and the distance from $\Theta_\mathrm B$ to the receiver as $x_2$, and using the convexity of $P(x)$, we have 
			\begin{align}\label{equ:ineq}
				&P(\lambda x_1+ (1-\lambda) x_2) <\lambda P(x_1) + (1-\lambda)P(x_2)\nonumber\\
				& <\max\left\{P(x_1), P(x_2)\right\}~ \forall \lambda\in(0,1) ~\text{and}~ x_1, x_2\in(0,W).
			\end{align}
			In (\ref{equ:ineq}),  $P(\lambda x_1+ (1-\lambda)x_2)$  parameterizes all path gains of the propagation paths via any cross street lying between $\Theta_\mathrm R$ and $\Theta_\mathrm B$, with different values of $\lambda$ selected. From the second inequality in \eqref{equ:ineq}, all these propagation paths have smaller path gain than that going through the streets specified in this proposition, which concludes the proof. 
		\end{proof}
		
	\end{proposition}
	Since the pathloss exponent of the segment $z_\mathrm P$ is $\alpha_\mathrm L$ and that of the segment $x_\mathrm P$ is $\alpha_\mathrm N$, with $\alpha_\mathrm L<\alpha_\mathrm N$, it is intuitive that the strongest path is more likely to be via the street \emph{closest to the receiver}, i.e., $\Theta_\mathrm R$. 
	
	To conclude the discussion on the uniqueness of the propagation path in the system model considered in this paper, we demonstrated that for both the typical and cross BSs, the propagation path is unique and also easy to identify based on the strongest path analysis above. For the parallel BS,  irrespective of whether the BS is located in the same block as the receiver, there are only two potential paths to be the strongest, and for analysis, we choose the path which traverses the cross street that is closest to the receiver.  
	\quad

	\section{Coverage Analysis}\label{sec:covana}
	The coverage probability serves as an important metric in evaluating system performance, since it is closely related to ergodic rate and throughput outage. In this section, we compute the coverage probability of a typical receiver in the MPLP microcellular network. First, we explain the independent thinning of the BSs considering the sectorized beam pattern of the mmWave BSs. Then, we analyze the CDF of the associated BS link gain based on the assumption that the receiver is associated to the closest BS (with smallest pathloss). In addition, we derive an accurate and concise  expression of the coverage probability. Finally, we examine the effects of the various components that contribute to interference mmWave microcellular networks. 
	\subsection{Independent thinning of BSs}\label{sec:thinning}
	Based on the sectorized antenna model in Section \ref{sec:antenna} and the properties of PPP, the BSs are independently thinned to generate two independent PPPs of BSs with antenna gain of $G$ and $g$ \cite{BacBla:Stochastic-geometry-and-wireless:09,BaiHea:Coverage-and-rate-analysis:15,Hae:Stochastic-geometry-for-wireless:12}. We define $p_\mathrm T$ as the thinning probability, and $\lambda_\mathrm B$ is the density of all active BSs deployed on the road side. After independent thinning, the densities of \emph{thinned} BSs with antenna gain of $G$ and $g$ are respectively $\lambda_\mathrm B p_\mathrm T$ and $\lambda_\mathrm B(1-p_\mathrm T)$. For the typical BSs, the thinning probability is $p_\mathrm T $ which equals to the probability that the receiver lies inside the main lobe, as defined in (\ref{equ:rbf}). For the cross BSs and the parallel BSs, we assume that only the BSs pointing towards the corner where the diffraction happens have beamforming gain as $G$. Hence, cross and parallel BSs have identical thinning probability as that of typical BSs. 
	To conclude, the thinning probabilities for three types of BSs (typical, cross and parallel) are identical, which are equal to the probability that the interfering BS has a beamforming gain of $G$
\begin{align}
p_\mathrm T  =  p,
\end{align}	
where $p$ is given in  (\ref{equ:rbf}).  The value of $p$ is hard to evaluate from a physical point of view, because propagation is dominantly down a street canyon, and it is dependent on the distribution of the interfering BS beam direction and multiple reflections along the street canyon. We can actually use any value for ``$p$" that occurs in practice. To make the exposition more clear, we pick the value as $p_\mathrm T = p = \frac{\theta}{2\pi}$, where $\theta$ is the beamwidth, under the assumption that the main lobe of the interfering BSs is uniformly distributed in the angular domain of $(0, 2\pi)$. 
	\subsection{Distribution of associated BS link gain}\label{sec:pathgaindist}
	To simplify SINR coverage analysis, we assume all links (association/interfering) experience independent and identically distributed (I.I.D.) Rayleigh fading with mean $1$, $h\sim \exp(1)$. We denote the normalized transmit power $P_\mathrm B = 1$ and represent the noise variance by $N_0$. Since the SINR expression in \eqref{(sinr_los)} is conditioned on the associated BS link gain $u$, we first analyze the distribution of $u$. Based on strongest BS association law, the receiver can be associated to either a typical/cross or parallel BS. The following lemma provides the cumulative density function (CDF) of the associated BS link gain of the typical/cross/parallel BS respectively. 
	
	\begin{lemma}\label{cdfanalysis1}
		The CDFs of the associated BS link gain of the typical BSs $u_1 = \max_{(x_\mathrm T\in\Phi_\mathrm T)} \{\ell_\mathrm L(x_\mathrm T)\}$, cross BSs $u_2= \max_{(x_\mathrm C, y_\mathrm C)\in\Phi_\mathrm C}\left\{\ell_\mathrm N (x_\mathrm C)\ell_\mathrm L(y_\mathrm C)\right\}$ and parallel BSs $u_3   =\max_{(x_\mathrm P, y_\mathrm P, z_\mathrm P)\in\Phi_\mathrm P}\{\ell_\mathrm N(x_\mathrm P)\ell_\mathrm N(y_\mathrm P)\ell_\mathrm L(x_\mathrm P)\}$ are approximated as 	\allowdisplaybreaks
		\begin{align}\label{equ:c3}
			F_{u_\mathrm T}(u) &= \exp\left(-\gamma_\mathrm 
			T		\lambda_\mathrm {B}u^{-\frac{1}{\alpha_{\mathrm L}}}\right), \nonumber\\
			F_{u_\mathrm C}(u)&   =\exp\left(-\gamma_{{\mathrm C}}\lambda_\mathrm {B}^\frac{\alpha_\mathrm L}{\alpha_\mathrm T} u^{-\frac{1}{\alpha_{\mathrm N}}}\right),\nonumber\\
			F_{u_\mathrm P}(u)&\approx  2\sqrt{2\gamma_\mathrm P \lambda_\mathrm S\lambda_\mathrm {B}^\frac{\alpha_\mathrm L}{\alpha_\mathrm N} u^{-\frac{1}{\alpha_\mathrm N}}}K_1\left(2\sqrt{2\gamma_\mathrm P\lambda_\mathrm S\lambda_\mathrm {B}^\frac{\alpha_\mathrm L}{\alpha_\mathrm N}u^{-\frac{1}{\alpha_\mathrm N}}}\right),
			\end{align}
		where 
		\begin{align}\label{equ:para}
			\gamma_\mathrm T&= 2G^\frac{1}{\alpha_\mathrm L}
			,\\
			\gamma_\mathrm C & = 2^{1+\frac{\alpha_\mathrm L}{\alpha_\mathrm N}}\lambda_\mathrm S(cG)^\frac{1}{\alpha_\mathrm N }\Gamma\left(1-\frac{\alpha_\mathrm L}{\alpha_\mathrm N}\right),\\
			\gamma_\mathrm P &= \gamma_\mathrm C c^\frac{1}{\alpha_\mathrm N},
				\end{align}
		and $K_1(\cdot)$ is the $1$-st order  modified Bessel's function of the second kind \cite{GraRyz:Table-of-integrals-series:14}. 
	\end{lemma}
	\begin{proof}
		See Appendix \ref{ap_a}. 
	\end{proof}

	Based on properties of the modified Bessel function, when the argument $\mu$ of $K_1(\mu)$ becomes small, it can be approximated as \cite{url_limiting} 
	\begin{align}\label{bessel_approx}
		K_1(\mu)\sim \mu^{-1}. 
	\end{align}	
	Since the argument of the modified Bessel function in (\ref{equ:c3}) scales with $\lambda_\mathrm S^2 \lambda_\mathrm{B}^\frac{\alpha_{\mathrm L}}{\alpha_{\mathrm N}}$, and the corner loss term $c$ further reduces the value, so that the approximation in (\ref{bessel_approx}) applies. Consequently, we can approximate (\ref{equ:c3}) as
	\begin{align}\label{equ:dist3}
		F_{u_\mathrm P}(u)&\approx 
2\sqrt{2\gamma_\mathrm P \lambda_\mathrm S\lambda_\mathrm {B}^\frac{\alpha_\mathrm L}{\alpha_\mathrm N} u^{-\frac{1}{\alpha_\mathrm N}}}\left(2\sqrt{2\gamma_\mathrm P \lambda_\mathrm S\lambda_\mathrm {B}^\frac{\alpha_\mathrm L}{\alpha_\mathrm N} u^{-\frac{1}{\alpha_\mathrm N}}}\right)^{-1}\nonumber \\&= 1 ,
	\end{align}	
	which implies that, generally, the largest gain from a parallel BS is small, i.e., the probability of being associated with a parallel BS is negligible.
	
	Using Lemma \ref{cdfanalysis1}, the CDF of the associated BS link gain $\mathrm U = \max\left\{u_\mathrm T, u_\mathrm C, u_\mathrm P\right\}$ can be evaluated as
	\begin{align}\label{equ:cdfv}
		&	F_{\mathrm U}(u) = \mathbb{P}\left(\max\{u_\mathrm T, u_\mathrm C, u_\mathrm P\}<u\right) \nonumber\\
		&\mathop{=}^{(a)} \mathbb{P}\left(\max\{u_\mathrm T\}<u\right) \mathbb{P}\left(\max\{u_\mathrm C\}<u\right) \mathbb{P}\left(\max\{u_\mathrm P\}<u\right)\nonumber\\
		&\mathop{\approx}^{(b)} \exp\left(-\gamma_\mathrm T\lambda_\mathrm {B} u^{-\frac{1}{\alpha_\mathrm L}}\right)\exp\left(-\gamma_\mathrm C\lambda_\mathrm {B}^\frac{\alpha_\mathrm L}{\alpha_\mathrm N} u^{-\frac{1}{\alpha_\mathrm N}}\right), 
	\end{align}
	where $(a)$ is based on the fact that the locations of the typical/cross/parallel BSs are mutually independent, $(b)$ follows the results of Lemma \ref{cdfanalysis1} and the observation that the association with parallel BSs is negligible.

	Fig. \ref{asso_dist_sl} compares the numerically evaluated CDF of the associated BS link gain for the following cases: i) association only with typical BSs, ii) with typical/cross BSs, and iii) considering all association cases, against the theoretical result given in (\ref{equ:cdfv}). 	
	The simulation parameters we use are summarized in Table \ref{simu_table}. The parameters are applicable to all of the following simulation results, unless stated otherwise. 
	\begin{table}[h!]
		\footnotesize
		\centering
		\caption{Simulation Parameters}
		\begin{tabular}{c|c}
			\hline
			\hline
			PARAMETERS & VALUES \\
			\hline
			\hline
			UPA Antenna Number $N_{\mathrm t}$ & $8\times8$ \\
			\hline
			\hline
			LOS Pathloss Exponent $\alpha_{\mathrm L}$ & $2.5$ \\
			\hline
			NLOS Pathloss Exponent $\alpha_{\mathrm N}$ & $7$ \\
			\hline
			Corner Loss  $\Delta$ & $20$dB \\
			\hline
			
			Intensity of Street	$\lambda_\mathrm{S}$ & 0.01 /m\\
			\hline 
			Intensity of BS $\lambda_\mathrm{B}$ & 0.01 /m\\
			\hline
			
		\end{tabular}
		\label{simu_table}
	\end{table}
It is seen that the analytic result matches well with the numerical result. It can also be seen that the empirical CDF curves obtained with and without the association with the parallel BSs coincide. This verifies the analysis in Lemma \ref{cdfanalysis1} and the subsequent approximation for largest gain seen by parallel BSs. Also, the curves show that the cross BSs association is also small compared to the typical BSs association, with the given simulation parameters.
	The simulation shows that LOS association with the typical BSs is dominant in the urban mmWave microcellular networks. 

	\begin{figure}
		\centering
		\includegraphics[width=5in]{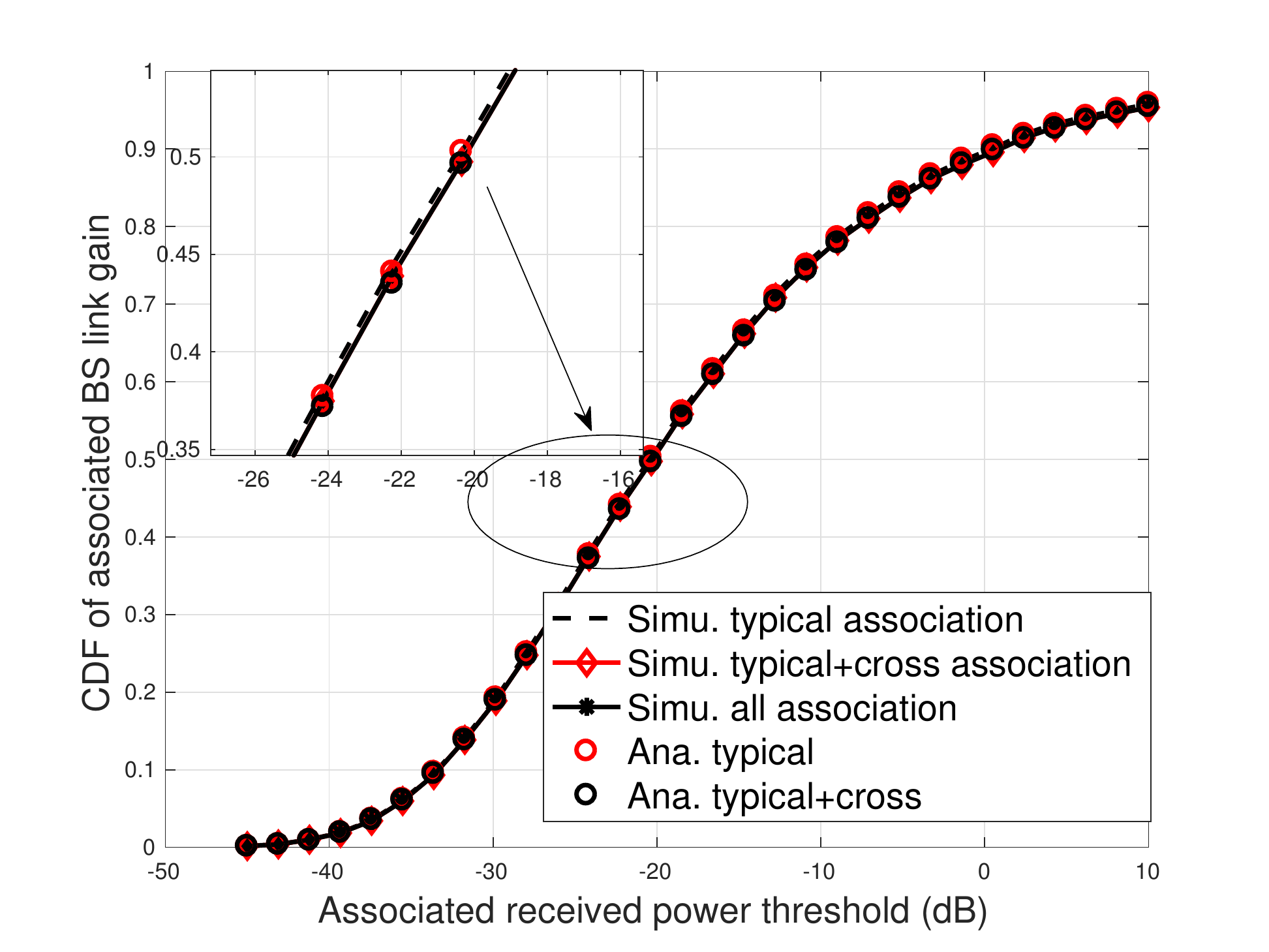}
		\caption{Comparison of analytic and numerical associated BS link gain
CDF. The black dashed line represents the association with only the typical BSs, the red solid line is the CDF considering association with both typical and cross BSs, and the black solid line is the result of considering all association cases. The red circle and black circle respectively denote the analytic result of CDF of associated BS link gain with only typical BS association and typical/cross BS association. }\label{asso_dist_sl}	
	\end{figure} 

	\subsection{Coverage probability}\label{subcov}
	In this section, we derive a closed-form expression for the coverage probability $p_\mathrm c(u, T)$ conditioned on the associated BS link gain as $u$. The coverage probability conditioned on $u$ is defined as 
	\begin{align}
		{p}_\mathrm c(u,  T) = \mathbb{P} \left(\mathrm{SINR} >T| u\right). \label{eqn:cov_p}
	\end{align}	
	Using \eqref{(sinr_los)} -- \eqref{i3}, \eqref{eqn:cov_p} can be expanded in terms of the Laplace transforms of interference conditioned on $u$ and noise as follows.
	\begin{align}\label{cov_prob_los1}
		&p_\mathrm c( u, T) 
		=\mathbb{P}\left(h>Tu^{-1}(N_0+I_{\phi_{\mathrm T}}(o) + I_{\phi_{\mathrm C}}(o)+I_{\phi_{\mathrm P}}(o))\right)\nonumber\\ &\mathop{=}^{(a)} \exp(-Tu^{-1}N_0)\mathcal{L}_{I_{\phi_{\mathrm T}}}(Tu^{-1})\mathcal{L}_{I_{\phi_{\mathrm C}}+I_{\phi_{\mathrm P}}}(Tu^{-1}),
	\end{align} 
	where $(a)$ is based on the assumption of I.I.D. Rayleigh fading channels, and $\mathcal{L}_{(\cdot)} $
	is the Laplace transform (LT) of random variable $(\cdot)$. Note that we cannot completely decouple the interference terms since the propagation links from the cross and parallel BSs could potentially share the same path segments, thus making their individual interference not independent. To analyze the problem, we start by examining the parallel BS interference. 
	
	\begin{proposition}\label{prop:nlos}
		The LT of the interference from the parallel BSs $I_{\phi_
			\mathrm P}$ is upper bounded by
		\begin{align}
			&	\mathcal{L}_{I_{\phi_\mathrm P}}\left(T, u\right)\gtrapprox 2\sqrt{2\gamma_\mathrm P\lambda_\mathrm S \lambda_\mathrm {B}^\frac{\alpha_\mathrm L}{\alpha_\mathrm N}\varrho(T)^\frac{\alpha_\mathrm L}{\alpha_\mathrm N} u^{-\frac{1}{\alpha_\mathrm N}}}\\
			&\times
			 K_1\left(2\sqrt{2\gamma_\mathrm P\lambda_\mathrm S\lambda_\mathrm {B}^\frac{\alpha_\mathrm L}{\alpha_\mathrm N}\varrho(T)^\frac{\alpha_\mathrm L}{\alpha_\mathrm N}u^{-\frac{1}{\alpha_\mathrm N}}}\right)\approx 1,
		\end{align}
		where $\gamma_{\mathrm P}$ is defined in (\ref{equ:para}), 
		and 
		\begin{align}\label{varrho}
			\varrho(t)& = \int_1^\infty \frac{1}{1+t^{-1}\mu^{\alpha_\mathrm L}}d\mu.
		\end{align}
		\begin{proof}
			The proof follows from the proof of Lemma \ref{cdfanalysis1} given in Appendix \ref{ap_a}, and is provided in Appendix \ref{proof:prop1}.
		\end{proof}
	\end{proposition}
	
	Since the lower bound of LT of the parallel interference evaluates to $1$ , which indicates that the interference from parallel BSs is small enough to be neglected, i.e., $I_{\phi_\mathrm P} \approx 0$.  Hence, the correlation of cross and parallel interference can be neglected and the coverage probability in (\ref{cov_prob_los1}) can be reformulated as 
	\begin{align}\label{cov_prob_los1}
		p_\mathrm c( u, T) 
		\approx \exp(-Tu^{-1}N_0)\mathcal{L}_{I_{\phi_{\mathrm T}}}(Tu^{-1})\mathcal{L}_{I_{\phi_{\mathrm C}}}(Tu^{-1}),
	\end{align} 
	which is derived in the following Theorem. 
	
	\begin{theorem}\label{cov_los}
		The coverage probability conditioned on the channel gain $u$ of the associated link is
		\begin{align}\label{cov_prob_los_exact}
			p_\mathrm c(u, T) &\approx \exp(-\beta_1 u^{-1})\exp(-\beta_2\lambda_\mathrm {B} u^{-\frac{1}{\alpha_{\mathrm L}}})\nonumber\\
			&\times\exp\left(-\beta_3\lambda_\mathrm {B}^\frac{\alpha_\mathrm L}{\alpha_\mathrm N}u^{-\frac{1}{\alpha_{\mathrm N}}}\right),\end{align}
		where
		\begin{align}\label{const_c}
			&\beta_1 = TN_0, \\
			&\beta_2  =\gamma_\mathrm T\left(p_\mathrm T  \varrho(T)+(1-p_\mathrm T)  \varrho\left(\frac{Tg}{G}\right)\right),  \nonumber\\
			&\beta_3  = \gamma_\mathrm C\left(p_\mathrm T^\frac{\alpha_\mathrm L}{\alpha_\mathrm N} \varrho(T)^\frac{\alpha_\mathrm L}{\alpha_\mathrm N}+(1-p_\mathrm T)^\frac{\alpha_\mathrm L}{\alpha_\mathrm N} \varrho\left(\frac{Tg}{G}\right)^\frac{\alpha_\mathrm L}{\alpha_\mathrm N}\right),
		\end{align}
		and $\varrho(\cdot)$ is defined in (\ref{varrho}). 
		\begin{proof}
			See Appendix \ref{appendix_theorem1}.
		\end{proof}
	\end{theorem}
	
	Using Theorem \ref{cov_los} and the distribution of the associated BS link gain in (\ref{equ:cdfv}), the SINR coverage probability can be evaluated as
	\begin{align}\label{equ:covProb}
		P_\mathrm{c}(T) = \int_0^\infty p_\mathrm c(u,T) f_{\mathrm U}(u)du,
	\end{align}
	where $p_\mathrm c(u, T)$ is provided in (\ref{cov_prob_los_exact}), and the probability density function (PDF) $f_\mathrm U(u)$ can be obtained from the derivative of the CDF derived in  (\ref{equ:cdfv}).

	\subsection{The effect of LOS and NLOS interferers}\label{int_effect}
	
	In Proposition \ref{prop:nlos}, we showed that the parallel BSs interference can be neglected in the analysis. In this section, we further compare the effects of typical interference $I_{\phi_{\mathrm T}}$ and cross interference $I_{\phi_{\mathrm C}}$.
	For tractable analysis, we assume the receiver is associated to the typical BS, so that we can have simpler associated BS link gain distribution. The analysis is based on the application of Jensen's inequality to the individual LT of $I_{\phi_{\mathrm T}}$ and $I_{\phi_{\mathrm C}}$. 
	
	From Theorem 1, by unconditioning the associated BS link gain $u$, the LT of the interference of BSs on the typical street is 
	$\mathcal{L}_{I_{\phi_{\mathrm T}}}(T) = \mathbb{E}_{u}\left[\exp\left(-\beta_2\lambda_\mathrm {B}u^{-	\frac{1}{\alpha_\mathrm L}}\right)\right]$
	and the LT of the interference due to the NLOS BSs on the cross streets is 
	$\mathcal{L}_{I_{\phi_{\mathrm C}}}(T) = \mathbb{E}_{u} \left[\exp\left(-\beta_3\left(\lambda_\mathrm {B}u^{-\frac{1}{\alpha_\mathrm L}}\right)^\frac{\alpha_\mathrm L}{\alpha_\mathrm N}\right)\right]$. Define two convex functions $\varphi_1(u) = \exp(-u)$ and $\varphi_2(u) = \exp(-u^\frac{\alpha_{\mathrm L}}{\alpha_{\mathrm N}})$. 		
	Since we assume the BS is associated to the typical BS in this case, the CDF of the associated BS link gain $u$ becomes  
	\begin{align}\label{equ:cdfu}
		F(u) = \exp\left(-\gamma_\mathrm T\lambda_\mathrm {B} u^{-\frac{1}{\alpha_\mathrm L}}\right).
	\end{align}
Based on the CDF of $u$ given above in (\ref{equ:cdfu}), we can derive the expectation of $u^{-\frac{1}{\alpha_\mathrm L}}$ as 
	\begin{align}
		\mathbb{E}_u\left[u^{-\frac{1}{\alpha_\mathrm L}}\right] = \frac{1}{\gamma_\mathrm T \lambda_\mathrm B}, 
	\end{align}
	hence, by Jensen's inequality, the lower bound of $\mathcal{L}_{I_{\phi_{\mathrm T}}}(T)$ becomes 
	\begin{align}\label{equ:jensen1}
	&\mathcal{L}_{I_{\phi_{\mathrm T}}}(T)\geq \mathcal{L}^{\text{LB}}_{I_{\phi_{\mathrm T}}}(T) = \exp\left(-\frac{\beta_2}{\gamma_\mathrm T}\right) \\&= \exp\left(-p_\mathrm T \varrho(T)-(1-p_\mathrm T) \varrho\left(\frac{Tg}{G}\right)\right). 
	\end{align} Similarly, we have derive the expectation of $(\lambda_\mathrm B u^{-\frac{1}{\alpha_\mathrm L}})^\frac{\alpha_\mathrm L}{\alpha_\mathrm N} = \lambda_\mathrm B^\frac{\alpha_\mathrm L}{\alpha_\mathrm N} u^{-\frac{1}{\alpha_\mathrm N}}$ as
	\begin{align}
		\mathbb{E}_u\left[(\lambda_\mathrm B u^{-\frac{1}{\alpha_\mathrm L}})^\frac{\alpha_\mathrm L}{\alpha_\mathrm N}\right] = \left(\frac{1}{\gamma_\mathrm T}\right)^\frac{\alpha_\mathrm L}{\alpha_\mathrm N}\Gamma\left(1+\frac{\alpha_\mathrm L}{\alpha_\mathrm N}\right), 
	\end{align}
	with the lower bound of $\mathcal{L}_{I_{\phi_{\mathrm C}}}(T)$ evaluated as 
	\begin{align}\label{equ:jensen2}
		&\mathcal{L}^{\text{LB}}_{I_{\phi_{\mathrm C}}}(T) = \exp\left(-\left(\frac{1}{\gamma_\mathrm T}\right)^\frac{\alpha_\mathrm L}{\alpha_\mathrm N}\beta_3 \Gamma\left(1+\frac{\alpha_\mathrm L}{\alpha_\mathrm N}\right)\right)\nonumber\\
		& = \exp\left(-2\lambda_\mathrm Sc^\frac{1}{\alpha_\mathrm N}\Gamma\left(1-\frac{\alpha_\mathrm L}{\alpha_\mathrm N}\right)\Gamma\left(1+\frac{\alpha_\mathrm L}{\alpha_\mathrm N}\right) \varepsilon\right). 
	\end{align}
	where we denote $\varepsilon$ as 
	\begin{align}\label{equ:varepsilon}
\varepsilon = \left(p_\mathrm T^\frac{\alpha_\mathrm L}{\alpha_\mathrm N} \varrho(T)^\frac{\alpha_\mathrm L}{\alpha_\mathrm N}+(1-p_\mathrm T)^\frac{\alpha_\mathrm L}{\alpha_\mathrm N} \varrho\left(\frac{Tg}{G}\right)^\frac{\alpha_\mathrm L}{\alpha_\mathrm N}\right).
	\end{align}
The argument inside the exponential function of (\ref{equ:jensen2}) scales with $\lambda_\mathrm S$ and $c^{\frac{1}{\alpha_\mathrm N}}$, where we have $\lambda_\mathrm S\ll1$ and $c\ll 1$. The inside argument is therefore effectively small. Another factor that might influence the cross street interference is the ratio between the pathloss exponents of the LOS/NLOS segments $r = \frac{\alpha_\mathrm L}{\alpha_\mathrm N}$. 

Generally, when $\alpha_\mathrm N$ is much larger than $\alpha_\mathrm L$, $\Gamma\left(1-\frac{\alpha_\mathrm L}{\alpha_\mathrm N}\right)\Gamma\left(1+\frac{\alpha_\mathrm L}{\alpha_\mathrm N}\right)$ is not that big, which leads to
	\begin{align}\label{equ:jensencomp}
		\mathcal{L}^{\text{LB}}_{I_{\phi_{\mathrm T}}}(T)\ll \mathcal{L}^{\text{LB}}_{I_{\phi_{\mathrm C}}}(T)\approx 1. 
	\end{align}
	The lower bound is shown to be fairly close to 1 but much larger than the lower bound of the LT of the typical interference. This indicates that in this case, the cross interference is much smaller than the typical interference, and also is negligible. 
	
	When $\alpha_\mathrm N$ is very close to $\alpha_\mathrm L$, however, $\Gamma\left(1-\frac{\alpha_\mathrm L}{\alpha_\mathrm N}\right)\Gamma\left(1+\frac{\alpha_\mathrm L}{\alpha_\mathrm N}\right)$ can be very large, which averages out the effects of small $\lambda_\mathrm S$ and $c^\frac{1}{\alpha_\mathrm N}$. In this case, cross interference can also contribute significantly in some certain urban canyons, where $\alpha_\mathrm N\to \alpha_\mathrm L$. This scaling law also leads to an intuitive insight that when the street intensity increases, the effects by cross BS interference grow larger. 

	Fig.~\ref{cov_prob_sl} gives a comparison between the analytic and simulation results of the coverage probability with different selections of $\alpha_\mathrm N$, and the cases when considering no interference (noise only), considering interference from only typical BSs, from both typical and cross BS interference, and from all interference. When $\alpha_\mathrm N = 7$, it is shown from the first five curves in the legend that the coverage probability curves with different interference components almost overlap. This verifies the corresponding proof in Proposition \ref{prop:nlos} that the parallel interference can be neglected, and the validity of Jensen's inequality lower bound analysis in (\ref{equ:jensen1})  and (\ref{equ:jensen2}). For  Fig. \ref{asso_dist_sl} and Fig. \ref{cov_prob_sl}, we set the corner loss as $\Delta = 20$dB. It will also be shown in Section \ref{sec:scaling} that with the corner loss ranging from $30$ dB to $0$ dB (no corner loss case), the coverage probability does not vary significantly. For the black curve pair, where $\alpha_\mathrm N = 2.51\to\alpha_\mathrm L = 2.5$, and the green curve pair $\alpha_\mathrm N= 2.52$,  there do exist certain differences between the coverage probability considering only the typical interference and considering also the cross interference, even though when $\alpha_\mathrm N = 2.52$, the difference becomes small already. We then choose the value $\alpha_\mathrm N = 3$, and it is shown that the coverage probability curves almost coincide with that of $\alpha_\mathrm N = 7$. Hence, we can conclude that under the Manhattan distance based pathloss model, $\alpha_\mathrm N$ does influence the contribution of cross street BS interference to the coverage probability. In most of the cases, the NLOS interference (from both cross and parallel BSs) is negligible; when and only when $\alpha_\mathrm N\to \alpha_\mathrm L$, the cross interference becomes significant enough to have an impact on the coverage probability. The effect of different selection of $\alpha_\mathrm N$ can also be observed from the parameter $\gamma_\mathrm C$ in (\ref{equ:para}), which scales with $\Gamma\left(1 - \frac{\alpha_\mathrm L}{\alpha_\mathrm N}\right)$, too. In the case when $\alpha_\mathrm N\to\alpha_\mathrm L$, the absolute value of the argument inside the exponential function representing the CDF of maximum cross BS power becomes large, which makes the CDF grow smaller, hence making it easier to be associated with a cross BS. For the following analysis, for ease of explanation, we adopt the pathloss exponent as $\alpha_\mathrm N = 7$, which is the value recorded the measurement results \cite{HanZhaTan:5G-3GPP-like-channel-models:16}, \cite{KarMolHur:Spatially-Consistent-Street-by-Street:17}.
	\begin{figure}[h!]
		\centering
		\includegraphics[width=5in]{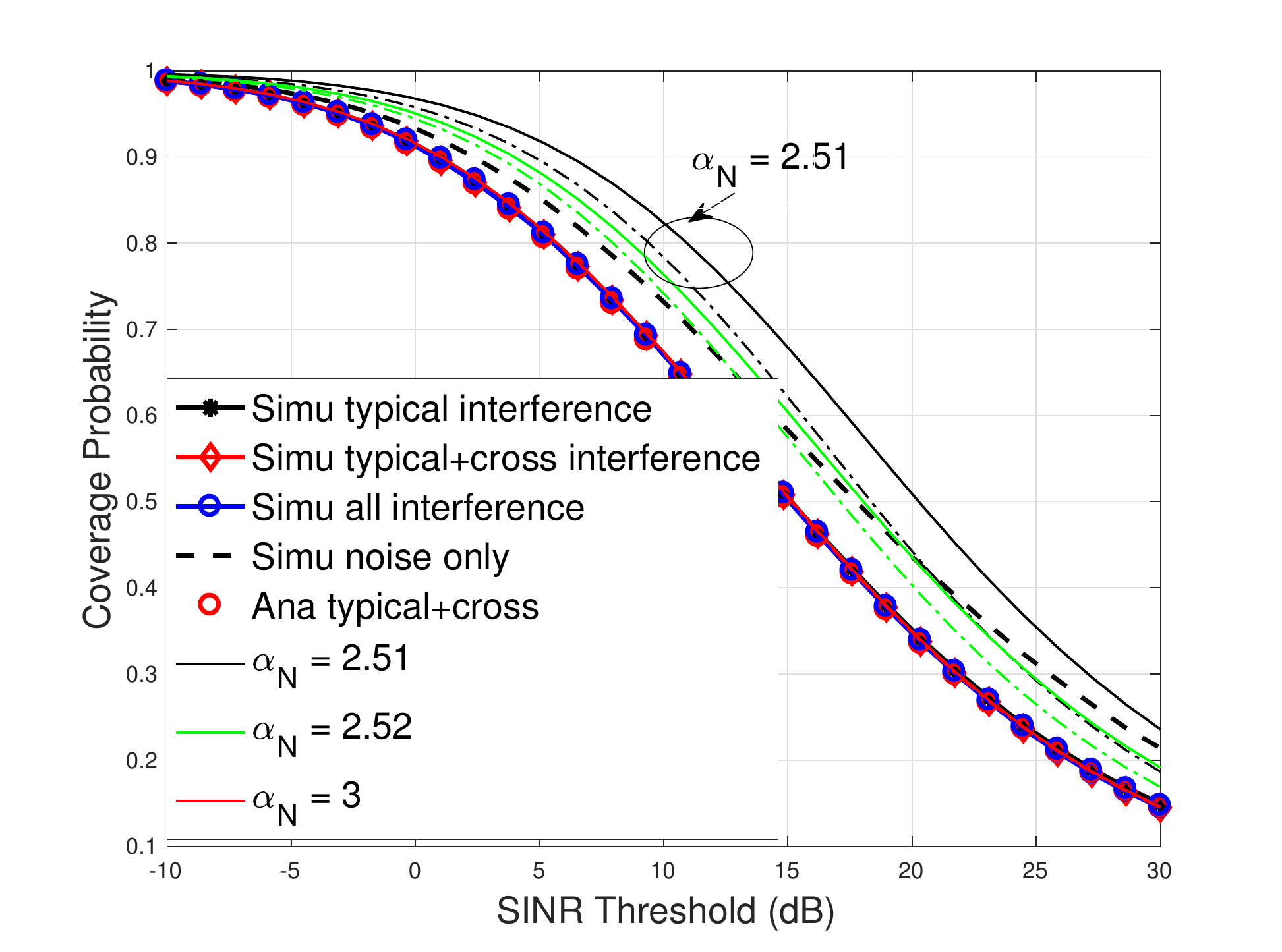}
		\caption{Comparison of the numerical and analytic coverage probability.
			The first five curves shown in the legend represent the result with $\alpha_\mathrm N = 7$. The first three solid lines represent the coverage probability considering only typical BS interference, both typical and cross BS interference and all interference. The black dashed line is the coverage probability simulated considering noise only. Red circles are the analytic expression of coverage probability in (\ref{cov_prob_los_exact}) -- (\ref{equ:covProb}) considering interference from typical and cross BSs.  The pairs of black/blue/red solid/dashed lines represent the coverage result with $\alpha_\mathrm N = 2.51, 2.52, 3$, as shown in the legend. The solid lines are the coverage probability with only typical interference and the dashed lines represent the case considering both typical and cross interference.}\label{cov_prob_sl}. 		
	\end{figure}
	\section{Scaling Laws with Network Densities }\label{sec:scaling}
	In this section, we analyze the scaling laws of the coverage probability and the association probability with the network densities, i.e., street intensity $\lambda_\mathrm S$ and BS intensity $\lambda_\mathrm B$. We apply tight approximations to the coverage probability and reveal interesting interplay between the performance and network deployment. 
	\subsection{Scaling laws for coverage probability}
	In this section, we focus on answering the following questions: i) how densely should BSs be deployed in  urban streets to maximize coverage at a minimum cost? ii) how does the coverage change for different densities in different cities?
	\subsubsection{Scaling law with BS intensity}\label{sec:interplay}
	The interference limited scenario targets an asymptotic case, where  the noise can be neglected and thus focus fully on the interplay between network intensities. This scenario can either be achieved by  high BS intensity (per street) or by dense streets deployment. 
	Based on the coverage probability given in (\ref{cov_prob_los_exact}) and (\ref{const_c}), after neglecting the noise term and changing variables $x =\lambda_{\mathrm{B}} u^{-\frac{1}{\alpha_{\mathrm L}}}$, the expression for the coverage probability becomes
	\begin{align}\label{equ:pc} 
	P_\mathrm{c}(T)&= \int_0^\infty \exp\left(-(\beta_2+ \gamma_\mathrm T)x\right)\exp\left(-(\beta_3+ \gamma_\mathrm C)x^{\frac{\alpha_{\mathrm L}}{\alpha_{\mathrm N}}}\right)\nonumber\\
	&\times\left(\gamma_\mathrm T+\frac{\gamma_\mathrm C \alpha_\mathrm L}{\alpha_\mathrm N}x^{\frac{\alpha_\mathrm L}{\alpha_\mathrm N}-1}\right)dx, 
	\end{align}
	where the parameters $\beta_2$, $\beta_3$, $\gamma_\mathrm T$ and $\gamma_\mathrm C$ are provided in Section \ref{sec:covana}.

	\begin{figure}
		\centering
		\setlength{\abovecaptionskip}{0pt}
		\setlength{\belowcaptionskip}{0pt}
		\includegraphics[width = 5in]{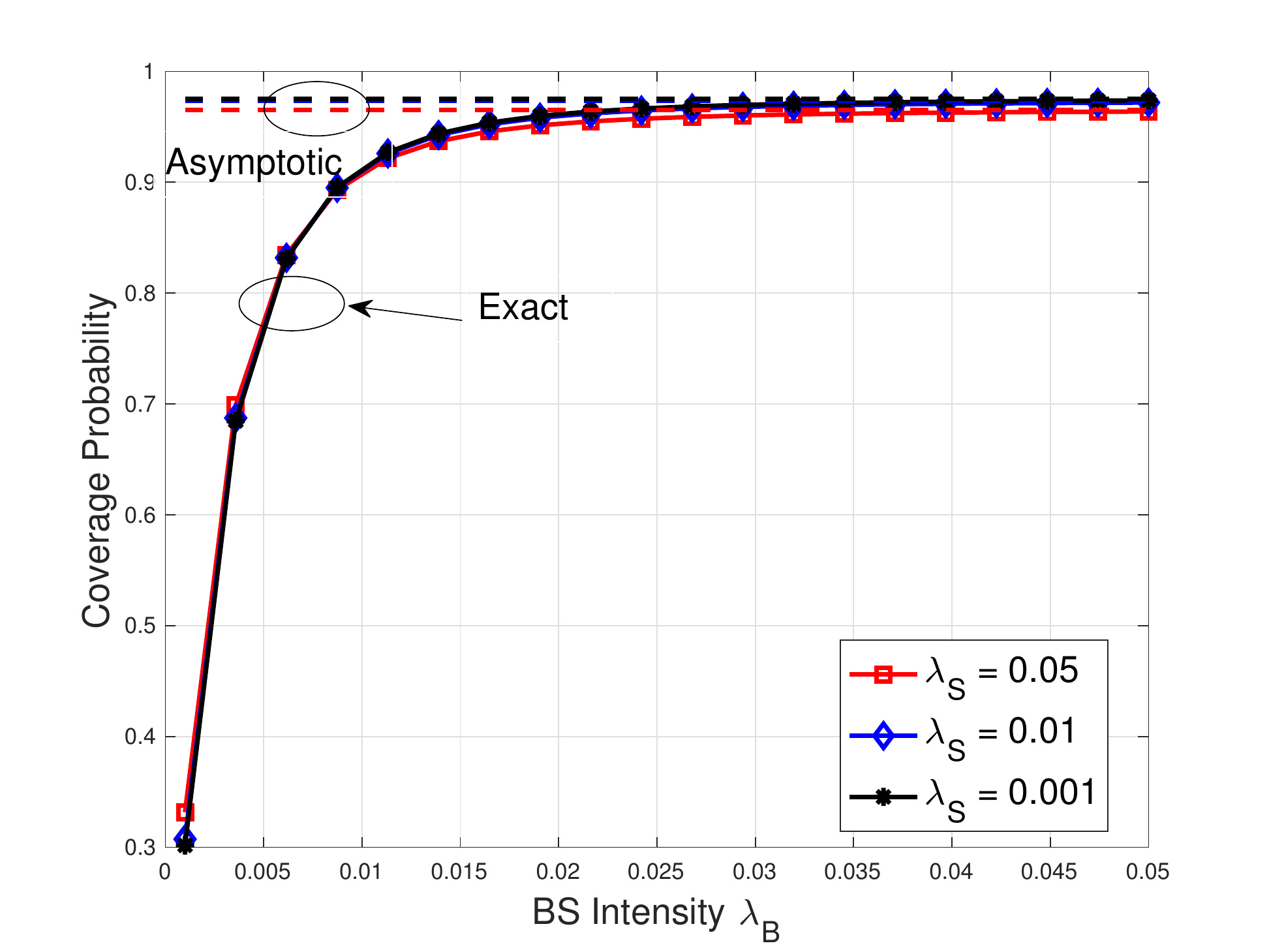}
		\caption{Asymptotic behavior of coverage probability with large BS intensity $\lambda_\mathrm b$. Solid red, green and blue curves are respectively the coverage probability under different street intensities, $\lambda_\mathrm S= 0.05, 0.01$ and $0.001$. Dashed curves represent the analytic asymptotic value of the coverage probability when BS intensity grows large.
		}\label{fig:scalinglt}
\end{figure}
	
	Under the Poisson models for BSs and the Manhattan distance pathloss model, one  observation from (\ref{equ:pc}) is that the coverage probability is independent of the BS intensity.  On one hand, when both street and BSs intensities grow large, it is intuitive that with ultra dense deployment of BSs, i.e., $\lambda_\mathrm B\to\infty$, both the associated link gain and interference become large. And therefore  their effects on the coverage probability cancel out, which leads to an asymptotic value of the coverage probability. On the other hand, when only the street intensity itself grows large, the scenario also becomes interference limited. In this case, the coverage probability is still a constant, however densely the BSs are deployed. This reveals an important insight that when street intensity grows large, the increase of coverage probability by deploying denser BSs is less significant. We plot Fig.~\ref{fig:scalinglt} to demonstrate the above two observations in an ultra-dense network where intensity of BS grows large. First, it is shown that from  approximately $\lambda_\mathrm B = 0.05$ (average BS spacing of $20$m) for different street intensities, the coverage probability starts to converge to the asymptotic value. Second, with denser street distribution (e.g.,  $\lambda_\mathrm S= 0.1$, red curve), the increase of coverage probability is less prominent. Also, denser street distribution leads to lower asymptotic coverage probability.

	\subsubsection{Scaling law with street intensity}\label{sec:street}
	In the last section, we demonstrated the impact of different city streets (with different intensities) on the coverage probability enhancement. Next, we reveal the scaling laws between the coverage probability and the urban street intensity. One important thing to note is that in the dense street case, the street intensity $\lambda_\mathrm S$ is not arbitrarily large, where the most dense streets might have at least $20$m average spacing between them, with $\lambda_\mathrm S = 0.05$. We provide the following proposition to quantify how the coverage probability changes under different street intensities and prove it herein.  
	\begin{proposition}
		
		1) When the BS intensity $\lambda_\mathrm B$ is large, the coverage probability decreases linearly with the street intensity $\lambda_{\mathrm{S}}$. 
		2) When $\lambda_\mathrm B$ is small, the coverage probability increases linearly with $\lambda_{\mathrm{S}}$. 
		\begin{proof}
			In terms of the linear scaling law and its dependence on the BS intensity, we provide the following steps of the proof:  
			\paragraph{Linear scaling law}
			First, from (\ref{cov_prob_los_exact}) -- (\ref{equ:covProb}), the coverage probability can be rewritten as 
			\begin{align}\label{equ:cov_prob_ls}
				P_\mathrm c(T) =P_1 + P_2,
			\end{align}
			where 
			\begin{align}\label{equ:p1}
				&	P_1= \int_0^\infty \exp\left(-\beta_1 u^{-1}\right)\exp\left(-(\beta_2+\gamma_\mathrm T)\lambda_\mathrm {B} u^{-\frac{1}{\alpha_\mathrm L}}\right) \nonumber\\
				&\times\exp\left(-(\beta_3  + \gamma_\mathrm C)\lambda_\mathrm {B}^\frac{\alpha_\mathrm L}{\alpha_\mathrm N}u^{-\frac{1}{\alpha_\mathrm N}}\right)\left(\frac{\lambda_\mathrm {B} \gamma_\mathrm T}{\alpha_\mathrm L}u^{-\frac{1}{\alpha_\mathrm L}-1}\right)du,
			\end{align}
			and 
			\begin{align}\label{equ:p2before}
				P_2 = \int_0^\infty \exp\left(-\beta_1 u^{-1}\right)\exp\left(-(\beta_2+\gamma_\mathrm T)\lambda_\mathrm {B} u^{-\frac{1}{\alpha_\mathrm L}}\right)\nonumber\\
				\times\exp\left(-(\beta_3  + \gamma_\mathrm C)\lambda_\mathrm {B}^\frac{\alpha_\mathrm L}{\alpha_\mathrm N}u^{-\frac{1}{\alpha_\mathrm N}}\right)\left(\frac{\gamma_\mathrm C}{\alpha_\mathrm N}\lambda_\mathrm {B}^\frac{\alpha_\mathrm L}{\alpha_\mathrm N}u^{-\frac{1}{\alpha_\mathrm N}-1}\right)du.
			\end{align}
			We then rewrite the second part in (\ref{equ:p2before}), using integration by parts, as
			\begin{align}\label{equ:p2after}
				& P_2 = \frac{\gamma_\mathrm C}{\gamma_\mathrm C+\beta_3}\int_0^\infty \exp\left(-\beta_1 u^{-1}\right)\exp\left(-(\beta_2+\gamma_\mathrm T)\lambda_\mathrm {B} u^{-\frac{1}{\alpha_\mathrm L}}\right)\nonumber\\
				&\times\frac{\partial \left[\exp\left(-(\beta_3  + \gamma_\mathrm C)\lambda_\mathrm {B}^\frac{\alpha_\mathrm L}{\alpha_\mathrm N}u^{-\frac{1}{\alpha_\mathrm N}}\right)\right]}{\partial u} \nonumber\\
				& =\frac{\gamma_\mathrm C}{\gamma_\mathrm C+\beta_3} - \frac{\gamma_\mathrm C}{\gamma_\mathrm C+\beta_3}\int_0^\infty \exp\left(-(\beta_3  + \gamma_\mathrm C)\lambda_\mathrm {B}^\frac{\alpha_\mathrm L}{\alpha_\mathrm N}u^{-\frac{1}{\alpha_\mathrm N}}\right)\nonumber\\
				&\times\frac{\partial\left[\exp\left(-\beta_1 u^{-1}\right)\exp\left(-(\beta_2+\gamma_1)\lambda_\mathrm {B} u^{-\frac{1}{\alpha_\mathrm L}}\right)\right]}{\partial u}.
			\end{align}
			In both \eqref{equ:p2after} and \eqref{equ:p1}, only $\beta_3 = \zeta_1 \lambda_\mathrm S$, and  $\gamma_\mathrm C = \zeta_2 \lambda_\mathrm S$ depend on $\lambda_\mathrm S$. Further, $\beta_3$ scales linearly with  $\gamma_\mathrm C$, which itself is small due to the terms $\lambda_\mathrm S$ and $c^{\frac{1}{\alpha_\mathrm N}}$. Then, by applying a first-order Taylor approximation $\exp(-x)\approx 1-x$ to 
			$\exp\left(-(\beta_3  + \gamma_\mathrm C)\lambda_\mathrm {B}^\frac{\alpha_\mathrm L}{\alpha_\mathrm N}u^{-\frac{1}{\alpha_\mathrm N}}\right) \approx 1  -\lambda_\mathrm s(\zeta_1+ \zeta_2)\lambda_\mathrm {B}^\frac{\alpha_\mathrm L}{\alpha_\mathrm N}u^{-\frac{1}{\alpha_\mathrm N}}$ in (\ref{equ:p1}) and (\ref{equ:p2after}), we can see $P_1$ and $P_2$ scale linearly with $\lambda_\mathrm S$, hence proving the linear scaling law of coverage probability with $\lambda_\mathrm S$. 	Fig.~\ref{fig:taylorcomparison} compares the exact coverage probability in (\ref{equ:cov_prob_ls}) and that with the Taylor  approximation. It is shown that under different street intensities $\lambda_\mathrm S = 0.001, 0.01, 0.02$, the exact results match well with the Taylor approximations. This verifies the accuracy of using Taylor approximation to prove the linear scaling law. Another observation here is that when the street density is relatively small, e.g., $\lambda_\mathrm S= 0.001$, the coverage probability is insensitive to the NLOS pathloss exponent $\alpha_\mathrm N$, since the coverage almost remains a constant with $\alpha_\mathrm N$ ranging from 3 to 10. When streets become dense, the coverage probability decreases faster with growing $\alpha_\mathrm N$.  This is consistent with the fact that $\alpha_\mathrm N$ only affects pathloss of the NLOS links, and NLOS BS is negligible in either association or interference. 
			\paragraph{Dependence on BS intensity}
			To demonstrate the different scaling laws of coverage probability with BS intensities, we segregate the components in (\ref{equ:pc}) which are dependent on $\lambda_\mathrm S$ in the integral, and define it as $\Upsilon(\lambda_\mathrm S)$, which is 
			\begin{align}
				&\Upsilon(\lambda_\mathrm S)= \exp\left(-\lambda_\mathrm S\left(\zeta_1 + \zeta_2\right)\lambda_\mathrm {B}^\frac{\alpha_\mathrm L}{\alpha_\mathrm N}u^{-\frac{1}{\alpha_{\mathrm N}}}\right)\nonumber\\
				&	\times\left(\frac{\lambda_\mathrm {B} \gamma_\mathrm T}{\alpha_\mathrm L}u^{-\frac{1}{\alpha_{\mathrm L}}-1}+\frac{\lambda_{\mathrm{s}}\zeta_2\lambda_\mathrm {B}^\frac{\alpha_\mathrm L}{\alpha_\mathrm N}}{\alpha_\mathrm N}u^{-\frac{1}{\alpha_{\mathrm N}}-1}\right), 
			\end{align}
			the derivative of which is 
			\begin{align}\label{equ:deri}
				&	\Upsilon'(\lambda_\mathrm S) = \frac{\lambda_\mathrm {B}^\frac{\alpha_\mathrm L}{\alpha_\mathrm N}}{\alpha_\mathrm N}u^{-\frac{1}{\alpha_\mathrm N}-1}\exp\left(-\lambda_\mathrm S\left(\zeta_1 + \zeta_2\right)\lambda_\mathrm {B}^\frac{\alpha_\mathrm L}{\alpha_\mathrm N}u^{-\frac{1}{\alpha_{\mathrm N}}}\right)\nonumber\\
				&\times
				\left(\zeta_2 - (\zeta_1 + \zeta_2)\alpha_\mathrm N\left[\frac{\gamma_\mathrm T \lambda_\mathrm {B}}{\alpha_\mathrm L}u^{-\frac{1}{\alpha_\mathrm L}} + \frac{\lambda_\mathrm S\zeta_2 \lambda_\mathrm {B}^\frac{\alpha_\mathrm L}{\alpha_\mathrm N}}{\alpha_\mathrm N}u^{-\frac{1}{\alpha_\mathrm N}}\right]\right). 
			\end{align}
			Since the exponential part from (\ref{equ:deri}) is always positive, and $\zeta_2$ and $\zeta_1$ are independent of $\lambda_\mathrm B$, there exists a threshold $\lambda_\mathrm B^*$, which satisfies
			\begin{align}
				\frac{\gamma_\mathrm T \lambda^*_\mathrm {B}}{\alpha_\mathrm L}u^{-\frac{1}{\alpha_\mathrm L}} + \frac{\lambda_\mathrm S\zeta_2 {\lambda^*_\mathrm {B}}^{\frac{\alpha_\mathrm L}{\alpha_\mathrm N}}}{\alpha_\mathrm N}u^{-\frac{1}{\alpha_\mathrm N}} = \frac{\zeta_2}{(\zeta_1+\zeta_2)\alpha_\mathrm N}. 
			\end{align} 	Hence, when $\lambda_\mathrm {B}>\lambda_\mathrm B^*$, $\Upsilon'(\lambda_\mathrm S)<0$, which indicates that when BS intensity is larger, coverage probability decreases with $\lambda_\mathrm S$. Further, when $\lambda_\mathrm B<\lambda_\mathrm B^*$, denser streets lead to a higher coverage probability. 		\end{proof}
	\end{proposition}
\begin{figure}
	\centering
			\includegraphics[width = 5in]{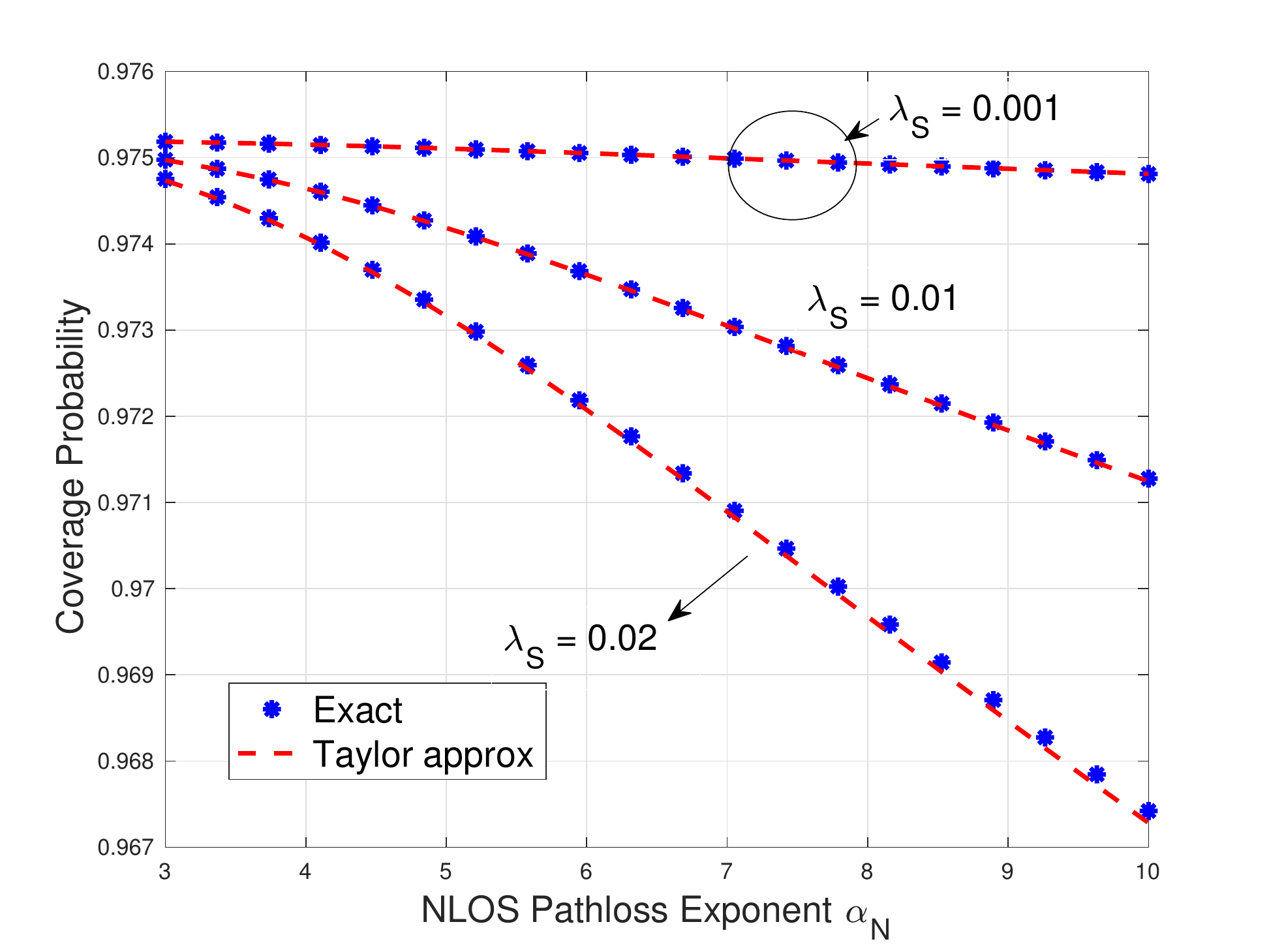}
			\caption{Comparison of the exact and Taylor approximation of coverage probability, with different NLOS pathloss exponent $\alpha_\mathrm N$. Blue star dots plot exact coverage probability in Theorem \ref{cov_los} under different street intensities, i.e., $\lambda_\mathrm S= 0.02,  0.01$ and $0.001$. The blue stars are the exact coverage probability and the red dashed lines are the Taylor approximations to (\ref{equ:p1}) and (\ref{equ:p2after}). }\label{fig:taylorcomparison}
\end{figure}
\begin{figure}
\centering
			\includegraphics[width =  5in]{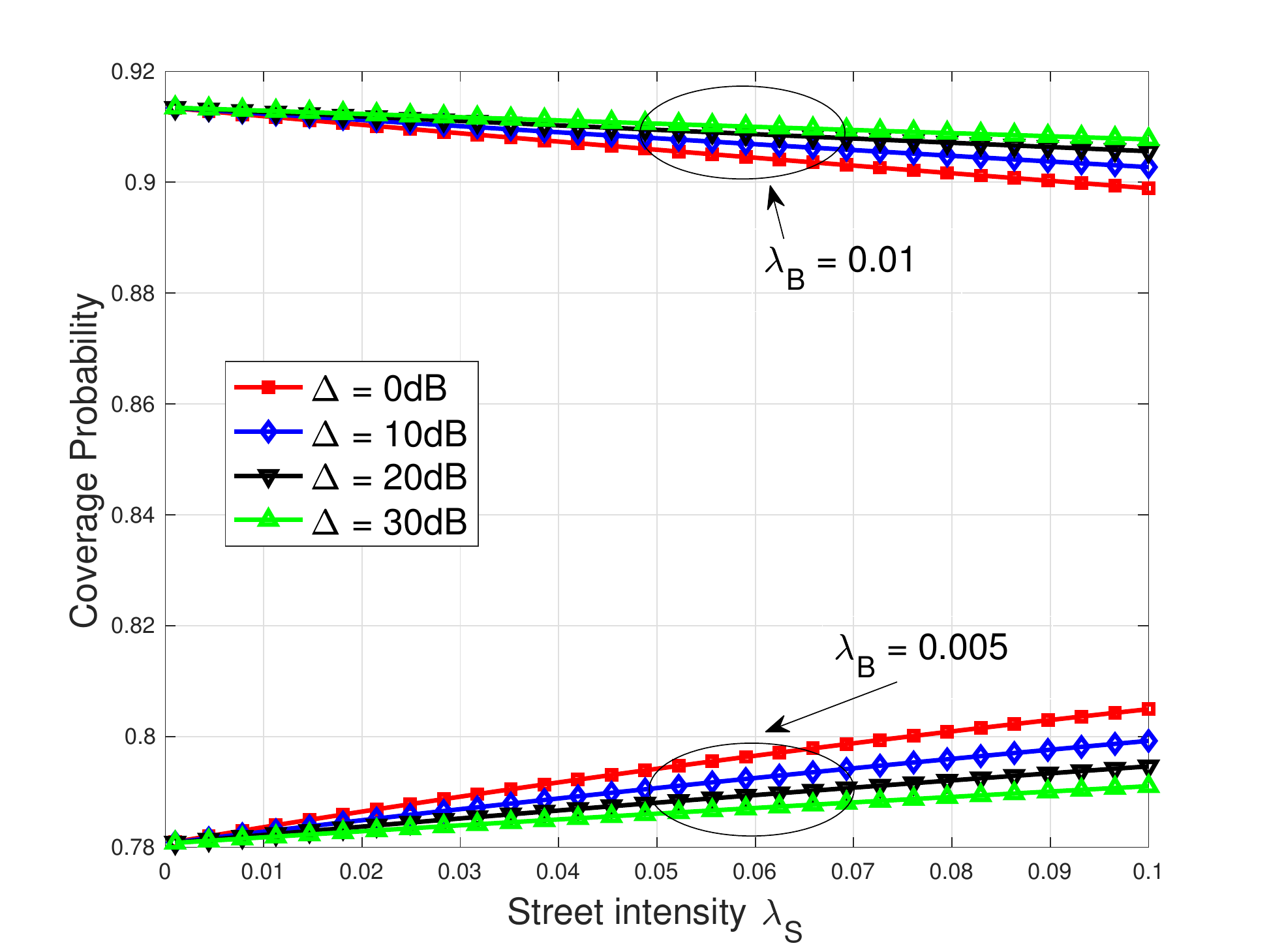}
		\caption{Scaling of coverage probability with different street intensities. Comparison is made under sparse/dense BS deployments $\lambda_\mathrm B = 0.005, 0.01$, and different corner losses, including no corner loss case, where $\Delta = 0$dB. 		}\label{sl_cov_scaling_ls}
\end{figure}

	Fig.~\ref{sl_cov_scaling_ls} illustrates the linear scaling of the coverage probability with the intensity of streets $\lambda_\mathrm S$. It first can be observed that the coverage probability scales linearly with the intensity of streets, and the coverage probability increases with $\lambda_\mathrm S$ while it decreases with  corner loss $\Delta$, when the BS intensity is relatively small $\lambda_\mathrm B = 0.005$. Also, the coverage probability decreases with $\lambda_\mathrm S$ with large BS intensity $\lambda_\mathrm B = 0.01$, while it increases with corner loss in the meantime. This implies that when the BS deployment is dense, interference becomes dominant and larger corner loss reduces the interference; when BSs are relatively sparse, small corner loss strengthens the signal from the cross BSs, thus making the associated link gain  stronger and enhancing the coverage probability. Also, it can be observed that when the corner loss becomes small (e.g., the no shadowing loss case $\Delta = 0$dB), the coverage probability becomes more sensitive to the change of street intensities, which is shown by a larger slope of the curve of coverage probability. This is because the smaller corner loss makes the cross BS interference more prominent, thus increasing the sensitivity of coverage probability to the street intensities.

	From the above analysis, the microcellular network does not work efficiently in a scenario where both BS and street intensities are large. When the BSs are sparsely deployed in an urban landscape with increasing street intensities (i.e., where blocks are small), then a typical UE is more likely to be associated with a BS on cross streets, and also can have a larger associated BS link gain. When $\lambda_\mathrm B$ grows large, however, the system becomes interference-limited, thus dense BS deployments in dense streets only contribute to more interference and lower the coverage probability. This sheds light on an import conclusion that ultra-dense BS deployment should be avoided in an urban canyon with dense street densities.

	\subsection{Scaling law for BS association probability }\label{sec_sl_los}

	In this section, we analyze the BS association under the Manhattan distance based pathloss model in MPLP. We start with the analysis of association probability. Given the CDF of the associated BS link gain in Section \ref{sec:pathgaindist}, we derive the probability the receiver is associated with a LOS BS on the typical street.
	\begin{corollary}\label{cor2}
		The probability $\chi_\mathrm {T}$ that the receiver is associated with a typical BS is 
		\begin{align}\label{prob_asso_nlos_sl}
			&\chi_\mathrm {T}\mathop{=}^{(a)}\mathbb{E}_u\left\{\mathbb{P}\left(u_\mathrm C<u \Big | u_\mathrm T = u \right)\right\} = \mathbb{E}_{u_\mathrm T}\left\{\mathbb{P}\left(u_\mathrm C<u_\mathrm T \right)\right\}\nonumber\\&\mathop{=}^{(b)}\int_0^\infty \exp\left(-\gamma_\mathrm C \lambda_\mathrm {B}^\frac{\alpha_\mathrm L}{\alpha_\mathrm N}u^{-\frac{1}{\alpha_\mathrm N}} - \gamma_\mathrm T \lambda_\mathrm {B} u^{-\frac{1}{\alpha_\mathrm L}}\right)\nonumber\\
			&\times\frac{\gamma_\mathrm T \lambda_\mathrm {B}}{\alpha_\mathrm L}u^{-\frac{1}{\alpha_\mathrm L}-1} du\nonumber  \\
			& \mathop{=}^{(c)} \gamma_\mathrm T\int_0^\infty \exp\left(-\gamma_\mathrm C x^{\frac{\alpha_\mathrm L}{\alpha_\mathrm N}}-\gamma_\mathrm T x\right)dx,
		\end{align}
		where $(a)$ is conditioned of maximum path gain of typical BSs is $u$, $(b)$ is based on the CDF of the maximum path gain of typical/cross BSs, $(c)$ follows by change of variables $x = \lambda_\mathrm {B} u^{-\frac{1}{\alpha_\mathrm L}}$. 
	\end{corollary}
	Since the argument of the second exponential function in (\ref{prob_asso_nlos_sl}) is the multiplication of $\lambda_\mathrm {S}$ and an additional attenuation of corner loss, the argument inside tends to be small. Similar to the approximation in Section \ref{sec:scaling}, we approximate the association probability by 
	\begin{align}\label{assoprobapprox}
		&\chi_\mathrm {T}^\text{Approx} =  \int_0^\infty \exp\left(-\mu\right)\left(1 - \frac{\zeta_2}{\gamma_\mathrm T^\frac{\alpha_\mathrm L}{\alpha_\mathrm N}}\lambda_\mathrm S\mu^{\frac{\alpha_\mathrm L}{\alpha_\mathrm N}}\right)d\mu\nonumber\\
		& = 1- \frac{2^{\frac{\alpha_\mathrm L}{\alpha_\mathrm N}+1}\gamma_\mathrm C}{\gamma_\mathrm T^\frac{\alpha_\mathrm L}{\alpha_\mathrm N}}\left[\mathrm{sinc}\left(\frac{\alpha_\mathrm L}{\alpha_\mathrm N}\right)\right]^{-1}\lambda_\mathrm S, 
	\end{align}
	where $\mathrm{sinc}(x) = \frac{\sin \pi x}{\pi x}$. Because the sinc function monotonously decreases with $x$ ($0<x<1$), the association probability with a typical BS decreases with $\alpha_\mathrm L$, while it  increases with $\alpha_\mathrm N$. Fig. \ref{fig:linearassociation} provides the comparison of the exact association probability in (51) and the approximation result in (52). The approximation in (52) is tight when there exists corner loss $\Delta = 20$ dB, while the gap increases when the corner loss increases. There exists a linear scaling law for the association probability with the street intensity in the scenarios with significant shadowing loss at corner, which is shown in Fig.~\ref{fig:linearassociation}. Also, different from $\alpha_\mathrm N$, which only impacts on the NLOS BS pathloss, the LOS pathloss exponent $\alpha_\mathrm L$ is involved in both the calculation of typical/cross BS pathloss. The decrease of typical association probability with larger $\alpha_\mathrm L$ implies that the LOS link pathloss is more sensitive to the changing pathloss exponents. Also, it is intuitive that the increase of $\alpha_\mathrm N$ enhances the association probability since it further attenuates the transmit signal from cross street BSs. It should be noted that it is meaningful to examine the interplay between the coverage probability and these exponents values, since the pathloss exponent in reality is not fixed (we extract two reasonable parameters for the ease of analysis in this paper), but is a random variable varying from streets to streets \cite{KarMolHur:Spatially-Consistent-Street-by-Street:17}. The interplay of pathloss exponents and typical BS association probability provides insight into BS association behaviors under various channel conditions of different urban canyons.

	In addition, from (\ref{assoprobapprox}) there is a linear scaling law of the typical BS association probability with the intensity of cross streets in Fig.~\ref{fig:linearassociation}. Also, it should be noted that with the corner shadowing loss, even in an extremely dense street network, e.g., $\lambda_\mathrm S= 0.1$, the association probability with typical BSs $\chi_{\mathrm {T}}$ is still greater than $0.7$. Only when in the case with no street corner loss, the association probability $\chi_\mathrm T$ decreases significantly with the street intensity $\lambda_\mathrm S$.  The above association probability analysis illuminates another important observation that considering shadowing loss at a reasonable value, cross BSs play a minor role in BS association under the Manhattan distance based microcellular pathloss model. Similar effects on coverage probability have been demonstrated in Section \ref{subcov}. 
	
	\begin{figure}[h]
		\centering
		\setlength{\abovecaptionskip}{0pt}
		\setlength{\belowcaptionskip}{0pt}
		\includegraphics[width = 5in]{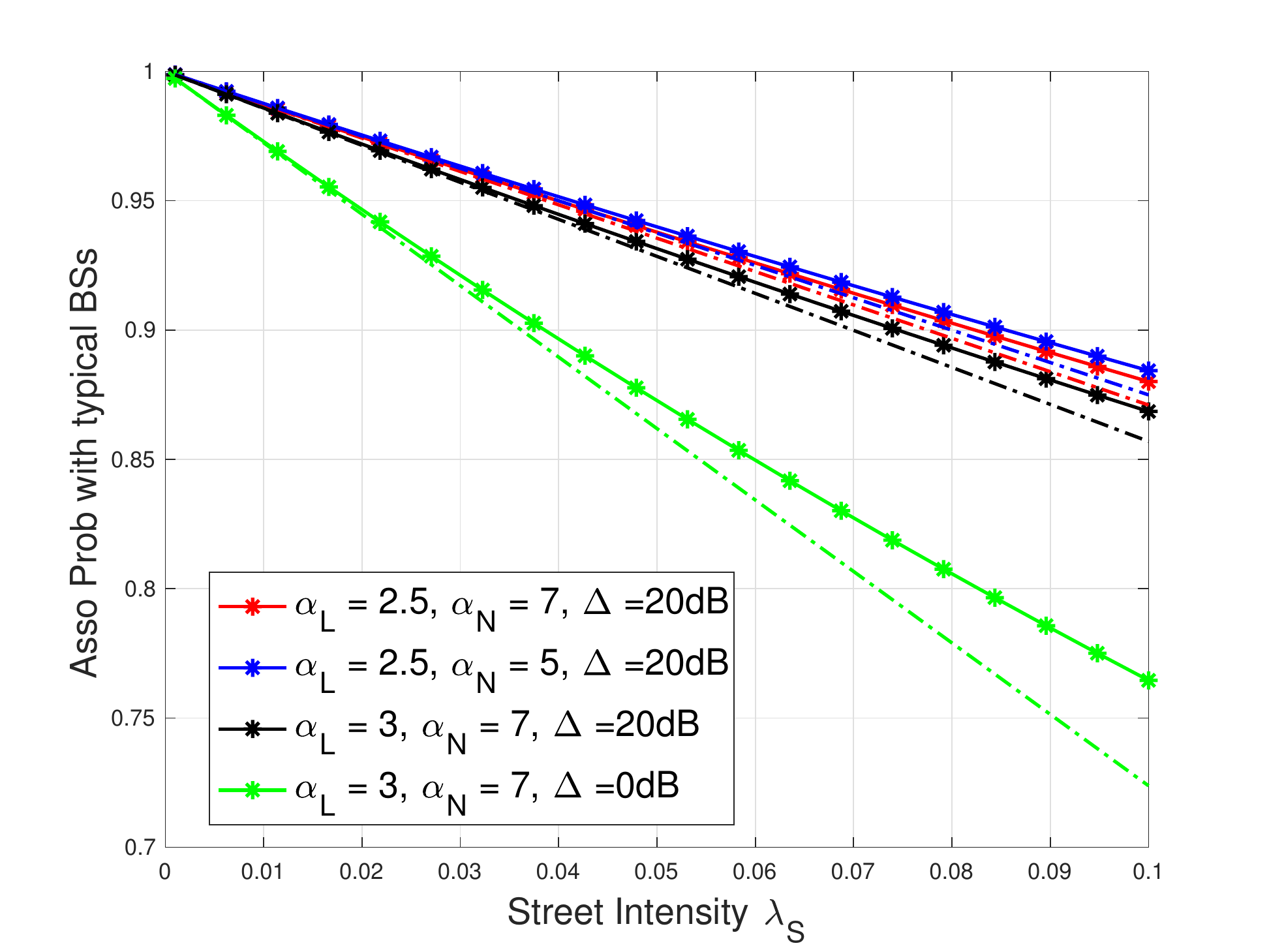}
		\caption{Illustration of the association probability with a typical BS. The solid lines represent the simulation result and the dashed lines are those obtained from the approximation in  (\ref{assoprobapprox}).}\label{fig:linearassociation}
	\end{figure}

	Hence, we can make  the following conclusions about the BS association. First, the BS association probability with the typical BSs is independent of the BS intensities.  Second, the association probability decreases linearly with the intensity of the cross streets. Third, typical BS association is less likely when the LOS pathloss exponent $\alpha_\mathrm L$ increases. 
	
	\section{Comparison of Different Street Models}
	In this section, we compare the ergodic rate under three different urban street models, the MPLP street modeling in this paper, fixed grid model (fixed spacing between streets) and realistic street deployments in Chicago, using the Manhattan distance pathloss model.  The ergodic rate is defined as $\mathcal{R} = \mathbb{E}\left[1 + \mathrm {SINR}\right]$. The raw street data is obtained by \emph{OpenStreetMap} powered by open source software and  \cite{url_openstreetmap}, \cite{ HakWeb:Openstreetmap:-User-generated-street:08}. We extract the map data by using GIS tool QGIS \cite{QGI:Quantum-GIS-geographic-information:11}. The simulated area is a region in Chicago given in  Fig.~\ref{fig:chicago}, and the extracted map which includes street and node information is plotted in Fig.~\ref{fig:matlabmap}.
	
	The parameters of the simulation scenario under the three street models are obtained based on the map we extracted from Chicago city. We assume all the street models have the same size ($1.659\times 2.002$ km$^2$). It can be counted from Fig. \ref{fig:chicago} that the number of the vertical and horizontal streets are respectively 15 and 8 (we only count \emph{main streets} which are shown explicitly in the map). Also, we assume the three models have the same (\emph{mean}) street numbers. These leads to the derivation of the horizontal street density as $\lambda_\mathrm {sh} \approx 4.8 /\text{km} $, and the vertical street density as $\lambda_\mathrm {sv} \approx 7.5 /\text{km} $. The densities are then applied to generate two independent PPPs for horizontal and vertical streets in the MPLP model, and set the spacing between two adjacent street respectively as $S_\text{h} = 133.5 \text{m}$ and $S_\text{v} =  207.4 \text{m}$. The comparison of the ergodic rate under the three models is given in Fig.~\ref{fig:capacitycompare}. From this figure, the ergodic rates are close under these different street models, which nearly coincide. The major reason for the observation is the negligible contribution of NLOS interference on the performance of Manhattan type mmWave microcellular networks from the analysis. The result, however, not only substantiates the negligibility of NLOS interference in MPLP networks, but  shows that the conclusion is also applicable to fixed grid and realistic urban canyons. Therefore, MPLP is an appropriate street model in understanding Manhattan type networks, which can yield simple yet accurate results and also provide interesting insights on the scaling of performance metrics. 
	\begin{figure}
		\centering
		\includegraphics[width=4in]{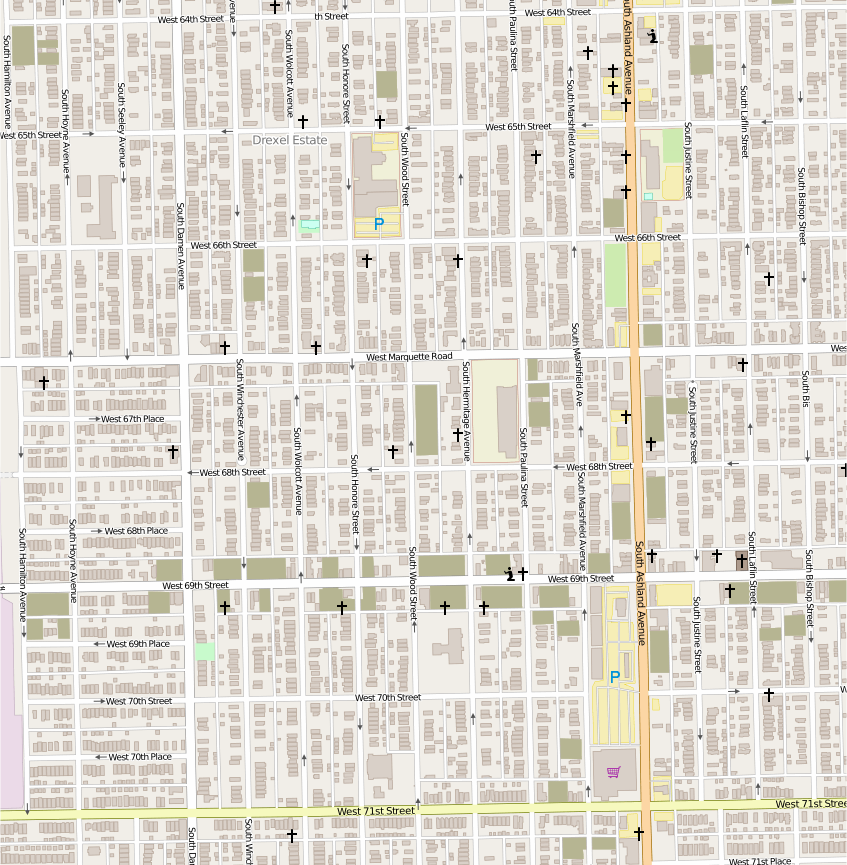}
		\caption{A snapshot of a region in Chicago from OpenStreetMap (Latitude: 41.762$\degree$N- 41.78$\degree$N, Longitude: -87.678$\degree$W - -87.658$\degree$W), with a size of  1.659$\times$2.002 ($\mathrm {km}^2$).}\label{fig:chicago}
		
	\end{figure}

	\section{Conclusion}\label{sec:conclusion}
	In this paper, we proposed a mathematical framework to model a Manhattan-type microcellular network under the urban mmWave communication system by stochastic geometry.  We first analyze the distribution of the path gain to the BS. We then derive an exact yet concise expression of the coverage probability. The LOS interference from the BSs on the same street as the serving BS is the dominating factor in determining the coverage probability, while BSs on cross and parallel streets have insignificant effects in most of the cases. We showed that in the ultra-dense network where intensity of BSs grows large, 
	\begin{figure}
		\centering
		\includegraphics[width=5in]{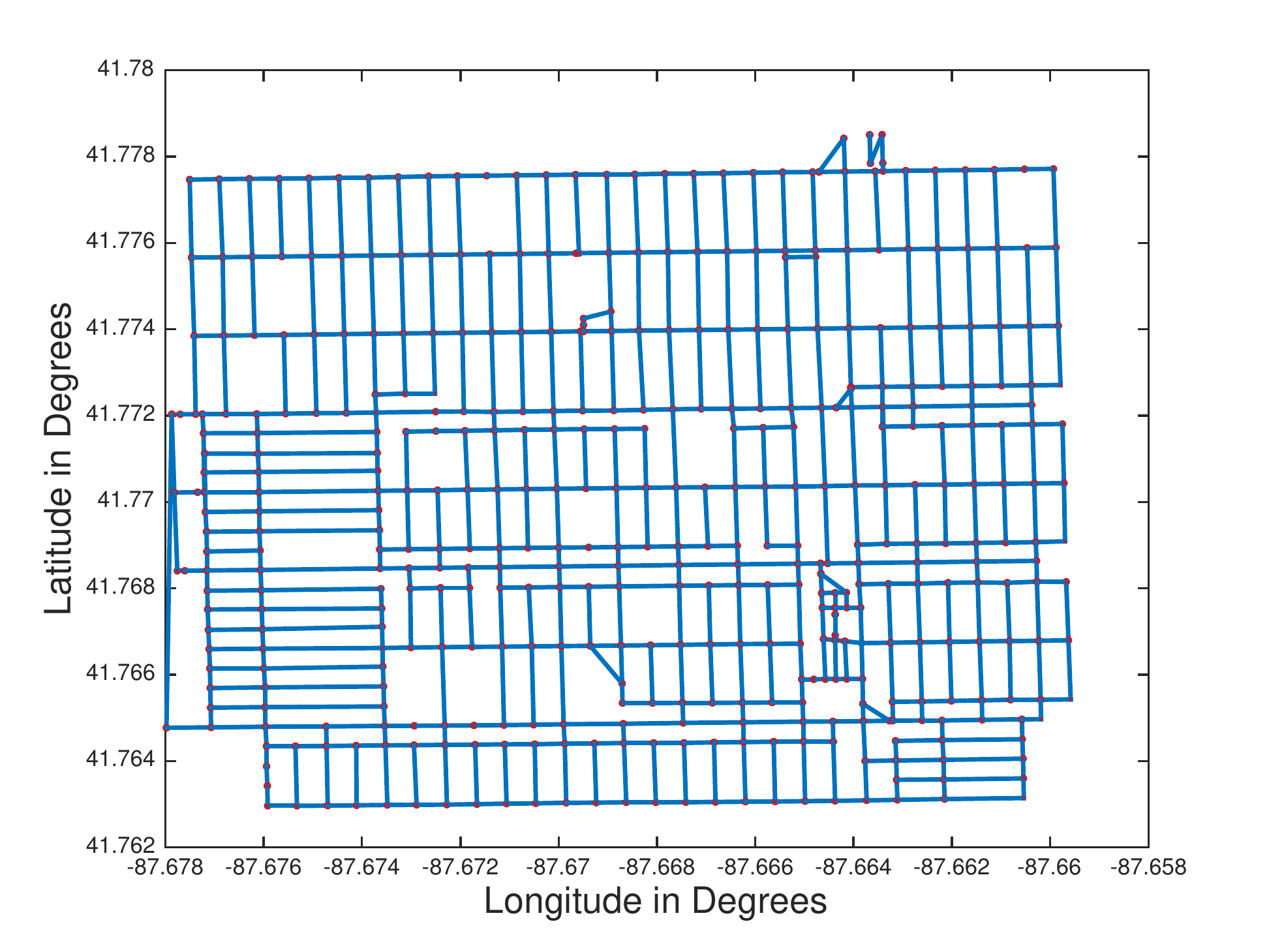}
		\caption{Streets abstracted from OpenStreetMap by QGIS. The red points are the intersections obtained from QGIS and the plot is obtained by lining up the intersections that has one common intersected street. }\label{fig:matlabmap}
	\end{figure}	
	\begin{figure}
		\centering
		\setlength{\abovecaptionskip}{0pt}
		\setlength{\belowcaptionskip}{0pt}
		\includegraphics[width = 5in]{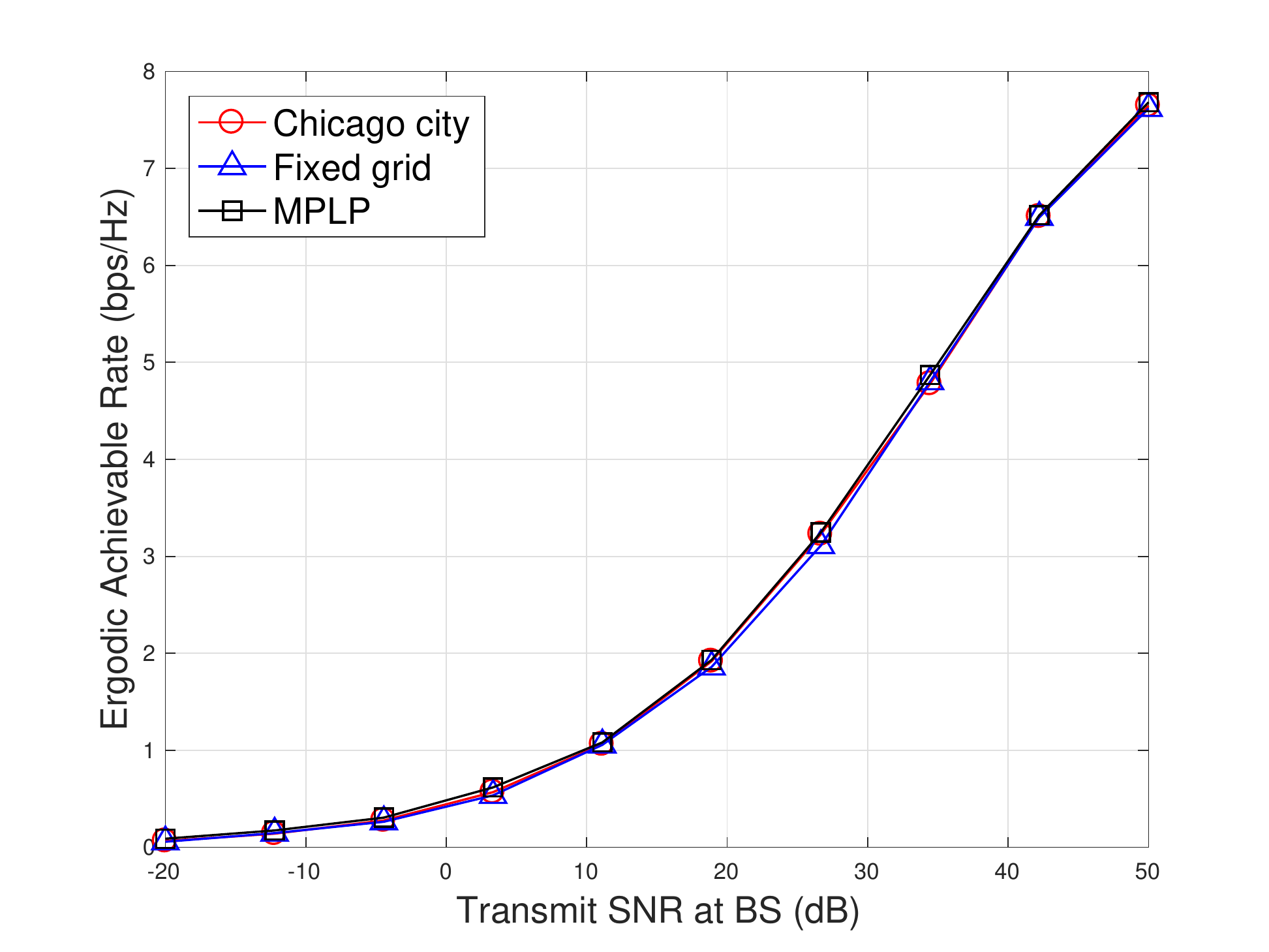}
		\caption{Comparison of ergodic rate among MPLP street modeling, fixed grid model and real streets obtained from Chicago city. 
		}\label{fig:capacitycompare}
	\end{figure}
	the network is interference-limited and the coverage probability approaches an asymptotic value.  Also, the coverage probability scales linearly with the intensity of streets, and displays an interesting interplay with the BS intensity: i) when BS deployment is dense, coverage probability decreases with street intensity; ii) when BS intensity is small, the coverage probability increases with street intensity. This implies that the system does not work efficiently when both BS and street intensities are large. Therefore, there is no need to deploy many BSs in an already dense urban street environment. In addition, we showed that in most of the cases,  the LOS BSs still dominate, from the perspective of both BS association, as well as coverage, unless in the case when $\alpha_\mathrm N \to\alpha_\mathrm L$. Finally, we numerically compared the ergodic rates under \emph{MPLP}, \emph{fixed spacing} and a \emph{realistic} street deployment in Chicago city. The ergodic rates under these street models match well, reinforcing the validity of MPLP as a realistic yet accurate urban street model in mmWave microcellular anlaysis.

	\appendices
	\section{Proof of Lemma \ref{cdfanalysis1}}\label{ap_a}
	
Since the receiver is always associated with the main lobe of the BS, which provides the smallest pathloss, the beamforming gain is always $G$. Hence, the CDF of the largest received power from the typical BSs is 
	\begin{align}
		&F_{u_\mathrm T}(u)  = \mathbb{P}\left(\max_{x\in\Phi_\mathrm T} Gx^{-\alpha_\mathrm L}<u\right)\nonumber\\
		& = \mathbb{P}\left(\min_{x\in\Phi_\mathrm T} x>G^{\frac{1}{\alpha_\mathrm L}}u^{-\frac{1}{\alpha_\mathrm L}}\right)\nonumber\\
		& \mathop{=}^{(a)} \exp\left(-2\lambda_\mathrm B  G^\frac{1}{\alpha_\mathrm L}u^{-\frac{1}{\alpha_\mathrm L}}\right) = \exp\left(-\gamma_\mathrm T\lambda_\mathrm B u^{-\frac{1}{\alpha_\mathrm L}}\right), 
	\end{align}
	where $(a)$ is based on the distribution of closest distance $x$ to one fixed point of one-dimensional PPP with intensity $\lambda$, and $\min\{x\}$ follows an exponential distribution, with parameter, $\min\{x\}\sim\exp\left(2\lambda\right)$, and also follows the independent thinning rule of BSs on the typical street of BSs with main lobe pointing to the receiver. 
	
	Similarly, the CDF of the largest received power from the BSs on the cross streets can be derived as follows 
	\allowdisplaybreaks[1]
	\begin{align}
		&F_{u_\mathrm C}(u)= \mathbb{E}_{\Phi_\mathrm C}\left[\prod^{(x_\mathrm C, y_\mathrm C\in\Phi_\mathrm C)} \mathbb{P}\left(x_\mathrm C^{-\alpha_{\mathrm N}}y_\mathrm C^{-\alpha_{ \mathrm L}}cG<u\right)\right]\nonumber\\
		&=\mathbb{E}_{x_\mathrm C}\left[\mathbb{E}_{y_\mathrm C}\left[\prod^{(x_\mathrm C, y_\mathrm C)\in\Phi_\mathrm C}\mathbb{P}\left(x_\mathrm C^{-\alpha_{\mathrm N}}\min(y_\mathrm C)^{-\alpha_{\mathrm L}}cG<u\right)\bigg | x_\mathrm C\right]\right]\nonumber\\
		& \mathop{=}^{(a)} \mathbb{E}_{x_\mathrm C}\left[\prod^{x_\mathrm C} \exp\left(-2\lambda_\mathrm B 
		x_\mathrm C^{-\frac{\alpha_{\mathrm N}}{\alpha_{\mathrm L}}}(cG)^{\frac{1}{\alpha_{\mathrm L}}}u^{-\frac{1}{\alpha_{\mathrm L}}}\right)\right]\nonumber\\
		& \mathop{=}^{(b)} \exp\left(-2\lambda_{\mathrm {S}}\int_0^\infty 1- \exp\left(-2\lambda_\mathrm B 
		x_\mathrm C^{-\frac{\alpha_{\mathrm N}}{\alpha_{\mathrm L}}}(cG)^{\frac{1}{\alpha_{\mathrm L}}}u^{-\frac{1}{\alpha_{\mathrm L}}}\right) dx\right)\nonumber\\
		& = \exp\left(-2\lambda_\mathrm S \left(2\lambda_\mathrm B \right)^\frac{\alpha_\mathrm L}{\alpha_\mathrm N}(cG)^\frac{1}{\alpha_\mathrm N}\Gamma\left(1-\frac{\alpha_\mathrm L}{\alpha_\mathrm N}\right)u^{-\frac{1}{\alpha_\mathrm N}}\right)\nonumber\\
		& = \exp\left(-\gamma_\mathrm C \lambda_\mathrm B^\frac{\alpha_\mathrm L}{\alpha_\mathrm N}u^{-\frac{1}{\alpha_\mathrm N}}\right),\end{align}
	where $(a)$ is derived by first conditioning on $x_\mathrm C$,  and $(b)$ is based on the probability generating functional (PGFL) of PPP. 
	
	Here, we provide an approximation of the CDF result of the associated BS link gain, based on the assumption that the strongest path is always via the cross street closest to the receiver, in Section \ref{sec:strongest}. The CDF can be derived as 
	\begin{align} & F_{u_\mathrm P}(u)\approx \mathbb{P}\left(\bigcap^{(x_\mathrm P,y_\mathrm P, z_\mathrm P )\in\Phi_\mathrm P} x_\mathrm P^{-\alpha_{\mathrm N}}y_\mathrm P^{-\alpha_{\mathrm N}}z_\mathrm P^{-\alpha_{\mathrm L}}c^2G<u\right) \nonumber\\
		& \mathop{=}^{(a)}  \mathbb{E}_{x_\mathrm P, y_\mathrm P}\left[\prod^{x_\mathrm P}\prod^{y_\mathrm P}  \exp\left(-2\lambda_\mathrm B G^{\frac{1}{\alpha_\mathrm L}} u^{-\frac{1}{\alpha_{\mathrm L}}}c^{\frac{2}{\alpha_{\mathrm L}}}x_\mathrm P^{-\frac{\alpha_{\mathrm N}}{\alpha_{\mathrm L}}}y_\mathrm P^{-\frac{\alpha_{\mathrm N}}{\alpha_{\mathrm L} }}\right)\right]\nonumber\\
			& \mathop{=}^{(b)}  \mathbb{E}_{x_\mathrm P }\left[\prod^{x_\mathrm P}\mathbb{E}_{y_\mathrm P}\left[\prod^{y_\mathrm P}  \exp\left(-2\lambda_\mathrm B\left( \frac{Gc^2}{ux_\mathrm P^{\alpha_\mathrm N}y_\mathrm P^{\alpha_\mathrm N}}\right)d^{\frac{1}{\alpha_\mathrm L}}  \right)\bigg| x_\mathrm P\right]\right]\nonumber\\
		& \mathop{=}^{(c)} \mathbb{E}_{x_\mathrm P} \left[\exp\left(-\frac{2\lambda_\mathrm S\left(\frac{Gc^2(2\lambda_\mathrm B)^{\alpha_\mathrm L}}{u}\right)^{\frac{1}{\alpha_\mathrm N}} \Gamma\left(1-\frac{\alpha_\mathrm L}{\alpha_\mathrm N}\right)}{x_\mathrm P}\right)\right]\nonumber\\
		& = \int_0^\infty 2\lambda_{\mathrm {S}} \exp\left(-\gamma_{\mathrm {P}}\lambda_\mathrm B^\frac{\alpha_\mathrm L}{\alpha_\mathrm N}u^{-\frac{1}{\alpha_\mathrm N}} x^{-1} - 2\lambda_{\mathrm {S}}x\right)dx\nonumber\\
		& \mathop{=}^{(c)} 2\sqrt{2\gamma_\mathrm P \lambda_\mathrm S\lambda_\mathrm {B}^\frac{\alpha_\mathrm L}{\alpha_\mathrm N} u^{-\frac{1}{\alpha_\mathrm N}}}K_1\left(2\sqrt{2\gamma_\mathrm P\lambda_\mathrm S\lambda_\mathrm {B}^\frac{\alpha_\mathrm L}{\alpha_\mathrm N}u^{-\frac{1}{\alpha_\mathrm N}}}\right), \end{align}
	where $\bigcup$ denotes the intersection of all of the events defined in the set $(x_\mathrm P,y_\mathrm P, z_\mathrm P )\in\Phi_\mathrm P$, $(a)$ is derived conditioned on $x_\mathrm P, y_\mathrm P$ and $(b)$ is derived conditioned on $x_\mathrm P$, $(c)$ is based on the PGFL function, $(d)$ follows the equation \cite{GraRyz:Table-of-integrals-series:14}, 
	\begin{align}
		\int_0^\infty \exp\left(-\frac{\beta}{4x} - \gamma x\right)dx = \sqrt{\frac{\beta}{\gamma}}K_1\left(\sqrt{\beta\gamma}\right). 
	\end{align}
	By simple calculations, we can conclude the proof. 
	\section{Proof of Proposition \ref{prop:nlos}}\label{proof:prop1}
		The derivation of the LT of the interference coming from the BSs on the parallel streets is similar. Based on the proof in Proposition \ref{prop:strongestpath}, the LT of the interference can be lower bounded by the LT assuming all the interfering beamforming gain is $G$, which is formulated as  
	\begin{align}  &\mathcal{L}_{I_{\phi_\mathrm P}} (Tu^{-1})\gtrapprox\nonumber\\ &\mathbb{E}_{\phi_\mathrm P}\left[\exp\left(-\sum_{(x_\mathrm P,y_\mathrm P, z_\mathrm P )\in\Phi_\mathrm P} Tu^{-1}hx_\mathrm P^{-\alpha_{\mathrm N}}y_\mathrm P^{-\alpha_{\mathrm N}}z_\mathrm P^{-\alpha_{\mathrm L}}c^2G\right)\right] \nonumber\\
	& \mathop{=}^{(a)}\mathbb{E}_{x_\mathrm P}\left\{\prod^{y_\mathrm P} \exp\left(-2\lambda_\mathrm B  G^\frac{1}{\alpha_\mathrm L}\varrho(T)u^{-\frac{1}{\alpha_\mathrm L}}c^\frac{2}{\alpha_\mathrm L}(xy)^{-\frac{\alpha_\mathrm N}{\alpha_\mathrm L}}\right)\right\}\nonumber\\
	&\mathop{=}^{(b)}\mathbb{E}_{x_\mathrm P}\left\{ \exp\left(-\gamma_\mathrm P \varrho(T)^\frac{\alpha_\mathrm L}{\alpha_\mathrm N}\lambda_\mathrm B^\frac{\alpha_\mathrm L}{\alpha_\mathrm N}u^{-\frac{1}{\alpha_\mathrm N}} x_\mathrm P^{-1}\right)\right\}\nonumber\\
	& = \int_0^\infty 2\lambda_{\mathrm {S}}  \exp\left(-\gamma_\mathrm P\varrho(T)^\frac{\alpha_\mathrm L}{\alpha_\mathrm N}\lambda_\mathrm B^\frac{\alpha_\mathrm L}{\alpha_\mathrm N}u^{-\frac{1}{\alpha_\mathrm N}} x^{-1} - 2\lambda_\mathrm S x\right)dx\nonumber\\
	& =2\sqrt{2\gamma_\mathrm P\lambda_\mathrm S \lambda_\mathrm {B}^\frac{\alpha_\mathrm L}{\alpha_\mathrm N}\left(\varrho(T)\right)^\frac{\alpha_\mathrm L}{\alpha_\mathrm N} u^{-\frac{1}{\alpha_\mathrm N}}}\nonumber\\&\times
K_1\left(2\sqrt{2\gamma_\mathrm P\lambda_\mathrm S\lambda_\mathrm {B}^\frac{\alpha_\mathrm L}{\alpha_\mathrm N}\left(\varrho(T)\right)^\frac{\alpha_\mathrm L}{\alpha_\mathrm N}u^{-\frac{1}{\alpha_\mathrm N}}}\right),  \end{align}
	where $(a)$ and $(b)$ follow the standard procedures in analysis of stochastic geometry and are similar to the proof of Laplace transform of $I_{\phi_\mathrm T}$ and $I_{\phi_\mathrm C}$ above in Appendix B. 
	
	\section{Proof of Theorem 1}\label{appendix_theorem1}
	
	We respectively give the LT of the three kinds of interferers $\phi_{\mathrm T}$, $\phi_{\mathrm C}$ and $\phi_{\mathrm P}$. The LT of the typical BS interference $\mathcal{L}^\mathrm {G}_{I_{\phi_\mathrm T}}(s)$ with beamforming gain as $G$ can be  given by 
	\allowdisplaybreaks
	\begin{align}\label{equ:proof1}
		&\mathcal{L}^\mathrm G_{I_{\phi_{\mathrm T}}}(s) = \mathbb{E}\left[\exp\left({-sG\sum_{x_\mathrm T \in \Phi_\mathrm T}} hx_\mathrm T^{-\alpha_\mathrm L}\right)\right]\nonumber\\
		& =\exp \left(-2\lambda_\mathrm{B}p_\mathrm T \int_{\left(\frac{u}{G}\right)^{-\frac{1}{\alpha_{\mathrm L}}}}^\infty \mathbb{E} \left(1-\exp\left(-sGhx_\mathrm T^{-\alpha_\mathrm L}\right)\right)\right)\nonumber\\
		& = \exp\left(-2\lambda_\mathrm{B}p_\mathrm T\int_{\left(\frac{u}{G}\right)^{-\frac{1}{\alpha_{\mathrm L}}}}^\infty\frac{1}{1+s^{-1}G^{-1}x_\mathrm T^{\alpha_\mathrm L}}dx_\mathrm T\right). 
	\end{align}
	For the interference with beamforming gain as $g$, the LT $\mathcal{L}^\mathrm {g}_{I_{\phi_\mathrm T}}(s)$  can be derived as 
	\begin{align}\label{equ:proof2}
&	\mathcal{L}^g_{I_{\phi_{\mathrm T}}}(s) =\mathbb{E}\left[\exp\left({-sg\sum_{x_\mathrm T \in \Phi_\mathrm T}} hx_\mathrm T^{-\alpha_\mathrm L}\right)\right]\nonumber\\
	& = \exp \left(-2\lambda_\mathrm{B}(1-p_\mathrm T) \int_{\left(\frac{u}{G}\right)^{-\frac{1}{\alpha_{\mathrm L}}}}^\infty \mathbb{E} \left(1-\exp\left(-sghx_\mathrm T^{-\alpha_\mathrm L}\right)\right)\right)\nonumber\\
		& = \exp\left(\int_{\left(\frac{u}{G}\right)^{-\frac{1}{\alpha_{\mathrm L}}}}^\infty\frac{-2\lambda_\mathrm{B}(1-p_\mathrm T)}{1+\left(\frac{Tg}{G}\right)^{-1}uG^{-1}x_\mathrm T^{\alpha_\mathrm L}}dx_\mathrm T\right). 
	\end{align}
By applying change of variables to (\ref{equ:proof1}) and (\ref{equ:proof2}), and combining the results above, the LT of the interference on the typical street can be formulated as 
	\begin{align}\label{equ:bigg}
	&	\mathcal{L}_{I_{\phi_{\mathrm T}}}(s) =	\mathcal{L}^\mathrm G_{I_{\phi_{\mathrm T}}}(s)	\mathcal{L}^\mathrm g_{I_{\phi_{\mathrm T}}}(s)\nonumber \\&=\exp\left(-\gamma_\mathrm T\lambda_\mathrm{B}u^{-\frac{1}{\alpha_{\mathrm L}}} \int_1^\infty \frac{1}{1+T^{-1}\mu^{\alpha_\mathrm L}}d\mu\right)\nonumber\\
		& = \exp\left(-\beta_2 \lambda_\mathrm Bu^{-\frac{1}{\alpha_\mathrm L}}\right). 
	\end{align}
	Similarly, the LT of the cross interfering with beamforming gain $G$ follows the the proof in Appendix A and proof of $\mathcal{L}^\mathrm G_{I_{\phi_{\mathrm T}}}(s)$, which can be given by 
	\begin{align}\label{equ:smallg}
	&	\mathcal{L}^\mathrm G_{I_{\phi_{\mathrm C}}}(s) =\mathbb{E}\left[\exp\left(-\sum_{(x_\mathrm C, y_\mathrm C)\in\Phi_\mathrm C}{sh x_\mathrm C^{-\alpha_{\mathrm N}}} y_\mathrm C^{-\alpha_{\mathrm L}} c G \underline{}\right)\right]\nonumber\\
		&=\mathbb{E}\left[\prod_{x_\mathrm C}\mathbb{E}\left[\prod_{y_\mathrm c}\exp\left(-shx_\mathrm C^{-\alpha_\mathrm N}y_\mathrm C^{-\alpha_\mathrm L}cG\right)\bigg|x_\mathrm C\right]\right]\nonumber\\
		& = \mathbb{E}\left[\prod_{x_\mathrm C}\exp\left(-2\lambda_\mathrm Bp_\mathrm T (cG)^\frac{1}{\alpha_\mathrm L}x^{-\frac{\alpha_\mathrm N}{\alpha_\mathrm L}}\varrho(T)\right)\right] \nonumber \\
&=		\exp\left(-2\lambda_\mathrm S (2\lambda_\mathrm B p_\mathrm T)^\frac{\alpha_\mathrm L}{\alpha_\mathrm N}\left(\frac{cG\varrho(T)^{\alpha_\mathrm L}}{u}\right)^\frac{1}{\alpha_\mathrm N}\Gamma\left(1-\frac{\alpha_\mathrm L}{\alpha_\mathrm N}\right)\right). 
	\end{align}
Combining the LT of the cross interference with beamforming gain $g$, the LT of the cross interference $\mathcal{L}_{I_{\phi_\mathrm C}}(u)$ derived accordingly. 

	\footnotesize
	\bibliographystyle{ieeetr}
	\bibliography{refer,IEEEabrv}

\begin{thebibliography}{10}

\bibitem{WanVenMol:Analysis-of-Urban-Millimeter:16}
Y.~Wang, K.~Venugopal, A.~F. Molisch, and R.~W. {Heath Jr.}, ``Analysis of
  urban millimeter wave microcellular networks,'' in {\em Proc. IEEE Veh.
  Technol. Conf. Fall (VTC Fall)}, pp.~1--5, Sept. 2016.

\bibitem{Mil:Vehicle-to-vehicle-to-infrastructure-V2V2I-intelligent:08}
J.~Miller, ``Vehicle-to-vehicle-to-infrastructure ({V2V2I}) intelligent
  transportation system architecture,'' in {\em Proc. IEEE Intell. Veh. Symp.},
  pp.~715--720, 2008.

\bibitem{BelValPai:On-wireless-links-for-vehicle-to-infrastructure:10}
P.~Belanovic, D.~Valerio, A.~Paier, T.~Zemen, F.~Ricciato, and C.~F.
  Mecklenbrauker, ``On wireless links for vehicle-to-infrastructure
  communications,'' {\em IEEE Trans. Veh. Technol.}, vol.~59, no.~1,
  pp.~269--282, 2010.

\bibitem{GozSepBau:IEEE-802.11p-vehicle:12}
J.~Gozalvez, M.~Sepulcre, and R.~Bauza, ``{IEEE} 802.11p vehicle to
  infrastructure communications in urban environments,'' {\em IEEE Commun.
  Mag.}, vol.~50, pp.~176--183, May 2012.

\bibitem{PiKha:An-introduction-to-millimeter-wave-mobile:11}
Z.~Pi and F.~Khan, ``An introduction to millimeter-wave mobile broadband
  systems,'' {\em IEEE Commun. Mag.}, vol.~49, no.~6, pp.~101--107, 2011.

\bibitem{RapGutBen:Broadband-millimeter-wave-propagation:13}
T.~S. Rappaport, F.~Gutierrez, E.~Ben-Dor, J.~N. Murdock, Y.~Qiao, and J.~I.
  Tamir, ``Broadband millimeter-wave propagation measurements and models using
  adaptive-beam antennas for outdoor urban cellular communications,'' {\em IEEE
  Trans. Antennas and Propag.}, vol.~61, no.~4, pp.~1850--1859, 2013.

\bibitem{RapHeaDan:Millimeter-wave-wireless:14}
T.~S. Rappaport, R.~W. Heath~Jr, R.~C. Daniels, and J.~N. Murdock, {\em
  Millimeter wave wireless communications}.
\newblock Pearson Education, 2014.

\bibitem{ChoVaGon:Millimeter-Wave-Vehicular-Communication:16}
J.~Choi, V.~Va, N.~Gonzalez-Prelcic, R.~Daniels, C.~R. Bhat, and R.~W. Heath,
  ``Millimeter-wave vehicular communication to support massive automotive
  sensing,'' {\em IEEE Commun. Mag.}, vol.~54, pp.~160--167, December 2016.

\bibitem{CheShaZhu:Infotainment-and-road-safety:11}
H.~T. Cheng, H.~Shan, and W.~Zhuang, ``Infotainment and road safety service
  support in vehicular networking: From a communication perspective,'' {\em
  Mechanical Systems and Signal Processing}, vol.~25, no.~6, pp.~2020--2038,
  2011.

\bibitem{AleSanCui:Measurement-and-Analysis-of-Propagation:08}
A.~V. Alejos, M.~G. Sanchez, and I.~Cuinas, ``Measurement and analysis of
  propagation mechanisms at 40 {GH}z: Viability of site shielding forced by
  obstacles,'' {\em IEEE Trans. Veh. Technol.}, vol.~57, pp.~3369--3380, Nov
  2008.

\bibitem{VaShiBan:Millimeter-Wave-Vehicular:16}
V.~Va, T.~Shimizu, G.~Bansal, and R.~W. Heath~Jr, ``Millimeter wave vehicular
  communications: A survey,'' {\em Foundations and Trends{\textregistered} in
  Networking}, vol.~10, no.~1, 2016.

\bibitem{NiuLiJin:A-survey-of-millimeter-wave:15}
Y.~Niu, Y.~Li, D.~Jin, L.~Su, and A.~V. Vasilakos, ``A survey of millimeter
  wave communications (mm{W}ave) for 5{G}: opportunities and challenges,'' {\em
  Wireless Networks}, vol.~21, pp.~2657--2676, Nov 2015.

\bibitem{AndBacGan:A-tractable-approach-to-coverage:11}
J.~G. Andrews, F.~Baccelli, and R.~K. Ganti, ``A tractable approach to coverage
  and rate in cellular networks,'' {\em IEEE Trans. Commun.}, vol.~59, no.~11,
  pp.~3122--3134, 2011.

\bibitem{BaiHea:Coverage-and-rate-analysis:15}
T.~Bai and R.~W. Heath, ``Coverage and rate analysis for millimeter-wave
  cellular networks,'' {\em IEEE Trans. Wireless Commun.}, vol.~14, no.~2,
  pp.~1100--1114, 2015.

\bibitem{BacZha:A-correlated-shadowing-model:15}
F.~Baccelli and X.~Zhang, ``A correlated shadowing model for urban wireless
  networks,'' in {\em Proc. IEEE Int. Conf. Comput. Commun. (INFOCOM)},
  pp.~801--809, 2015.

\bibitem{AndBaiKul:Modeling-and-Analyzing-Millimeter:17}
J.~G. Andrews, T.~Bai, M.~N. Kulkarni, A.~Alkhateeb, A.~K. Gupta, and R.~W.
  Heath, ``Modeling and analyzing millimeter wave cellular systems,'' {\em IEEE
  Trans. Commun.}, vol.~65, pp.~403--430, Jan 2017.

\bibitem{ElsKulBoc:Downlink-and-uplink-cell:16}
H.~Elshaer, M.~N. Kulkarni, F.~Boccardi, J.~G. Andrews, and M.~Dohler,
  ``Downlink and uplink cell association with traditional macrocells and
  millimeter wave small cells,'' {\em IEEE Trans. Wireless Commun.}, vol.~15,
  no.~9, pp.~6244--6258, 2016.

\bibitem{Di-:Stochastic-geometry-modeling:15}
M.~Di~Renzo, ``Stochastic geometry modeling and analysis of multi-tier
  millimeter wave cellular networks,'' {\em IEEE Trans. Wireless Commun.},
  vol.~14, no.~9, pp.~5038--5057, 2015.

\bibitem{BaiVazHea:Analysis-of-blockage-effects:14}
T.~Bai, R.~Vaze, and R.~W. Heath, ``Analysis of blockage effects on urban
  cellular networks,'' {\em IEEE Trans. Wireless Commun.}, vol.~13, no.~9,
  pp.~5070--5083, 2014.

\bibitem{KulSinAnd:Coverage-and-rate-trends:14}
M.~N. Kulkarni, S.~Singh, and J.~G. Andrews, ``Coverage and rate trends in
  dense urban mmwave cellular networks,'' in {\em Proc. IEEE Global Commun.
  Conf. (GLOBECOM)}, pp.~3809--3814, 2014.

\bibitem{LeeZhaBac:A-3-D-Spatial-Model-for-In-Building:16}
J.~Lee, X.~Zhang, and F.~Baccelli, ``A 3-{D} spatial model for in-building
  wireless networks with correlated shadowing,'' {\em IEEE Trans. Wireless
  Commun.}, vol.~15, pp.~7778--7793, Nov 2016.

\bibitem{MolKarWan:Millimeter-wave-channels-in-urban:16}
A.~F. Molisch, A.~Karttunen, R.~Wang, C.~U. Bas, S.~Hur, J.~Park, and J.~Zhang,
  ``Millimeter-wave channels in urban environments,'' in {\em Proc. European
  Conf. Antennas Propag. (EuCAP)}, pp.~1--5, 2016.

\bibitem{RapSunMay:Millimeter-wave-mobile:13}
T.~S. Rappaport, S.~Sun, R.~Mayzus, H.~Zhao, Y.~Azar, K.~Wang, G.~N. Wong,
  J.~K. Schulz, M.~Samimi, and F.~Gutierrez, ``Millimeter wave mobile
  communications for 5{G} cellular: It will work!,'' {\em IEEE Access}, vol.~1,
  pp.~335--349, 2013.

\bibitem{MacZhaNie:Path-loss-models:13}
G.~R. MacCartney, J.~Zhang, S.~Nie, and T.~S. Rappaport, ``Path loss models for
  5{G} millimeter wave propagation channels in urban microcells,'' in {\em
  Proc. IEEE Global Commun. Conf. (GLOBECOM)}, pp.~3948--3953, 2013.

\bibitem{NurKarRoi:METIS-channel-models:15}
V.~Nurmela, A.~Karttunen, A.~Roivainen, L.~Raschkowski, T.~Imai,
  J.~Jarvelainen, J.~Medbo, J.~Vihriala, J.~Meinila, K.~Haneda, {\em et~al.},
  ``{METIS} channel models,'' {\em FP7 METIS, Deliverable D}, vol.~1, 2015.

\bibitem{HanZhaTan:5G-3GPP-like-channel-models:16}
K.~Haneda, J.~Zhang, L.~Tan, G.~Liu, Y.~Zheng, H.~Asplund, J.~Li, Y.~Wang,
  D.~Steer, C.~Li, {\em et~al.}, ``5{G} 3{GPP}-like channel models for outdoor
  urban microcellular and macrocellular environments,'' in {\em Proc. IEEE Veh.
  Technol. Conf. Spring (VTC Spring)}, pp.~1--7, 2016.

\bibitem{KarMolHur:Spatially-Consistent-Street-by-Street:17}
A.~Karttunen, A.~F. Molisch, S.~Hur, J.~Park, and C.~J. Zhang, ``Spatially
  consistent street-by-street path loss model for 28-{GH}z channels in micro
  cell urban environments,'' {\em IEEE Trans. Wireless Commun.}, vol.~16,
  no.~11, pp.~7538--7550, 2017.

\bibitem{VenValHea:Device-to-Device-Millimeter-Wave:16}
K.~Venugopal, M.~C. Valenti, and R.~W. Heath, ``Device-to-device millimeter
  wave communications: Interference, coverage, rate, and finite topologies,''
  {\em IEEE Trans. Wireless Commun.}, vol.~15, pp.~6175--6188, Sept 2016.

\bibitem{url_openstreetmap}
{\url{http://www. openstreetmap.org/}},.

\bibitem{HakWeb:Openstreetmap:-User-generated-street:08}
M.~Haklay and P.~Weber, ``Openstreetmap: User-generated street maps,'' {\em
  IEEE Pervasive Comput.}, vol.~7, no.~4, pp.~12--18, 2008.

\bibitem{RamTopChi:OpenStreetMap:-using-and-enhancing:11}
F.~Ramm, J.~Topf, and S.~Chilton, {\em OpenStreetMap: using and enhancing the
  free map of the world}.
\newblock UIT Cambridge Cambridge, 2011.

\bibitem{QGI:Quantum-GIS-geographic-information:11}
D.~QGIS, ``Quantum {GIS} geographic information system,'' {\em Open Source
  Geospatial Foundation Project}, vol.~45, 2011.

\bibitem{Ber:A-recursive-method-for-street:95}
J.~E. Berg, ``A recursive method for street microcell path loss calculations,''
  in {\em Proc. IEEE Int. Symp. Pers. Indoor and Mobile Radio Commun. (PIMRC)},
  vol.~1, pp.~140--143 vol.1, Sep 1995.

\bibitem{Bal:Antenna-Theory:-Analysis:16}
C.~A. Balanis, {\em Antenna Theory: Analysis and Design}.
\newblock John Wiley \& Sons, 2016.

\bibitem{BacBla:Stochastic-geometry-and-wireless:09}
F.~Baccelli and B.~Blaszczyszyn, {\em Stochastic geometry and wireless
  networks: Theory}, vol.~1.
\newblock Now Publishers Inc, 2009.

\bibitem{Hae:Stochastic-geometry-for-wireless:12}
M.~Haenggi, {\em Stochastic geometry for wireless networks}.
\newblock Cambridge University Press, 2012.

\bibitem{GraRyz:Table-of-integrals-series:14}
I.~S. Gradshteyn and I.~M. Ryzhik, {\em Table of integrals, series, and
  products}.
\newblock Academic press, 2014.

\bibitem{url_limiting}
{\url{http://dlmf.nist.gov/10.30#SS1.info}},.

\end{thebibliography}
\end{document}